\documentclass[11pt]{amsart}
\pdfoutput=1

\usepackage[utf8]{inputenc}
\usepackage[margin=3cm]{geometry}
\usepackage{amssymb,amsmath,amsthm,amsfonts}
\usepackage{mathtools}
\usepackage{mathrsfs}
\usepackage{bm}
\usepackage{tikz}
\usepackage{tikz-cd}
\usepackage{hyperref}
\hypersetup{hidelinks}
\usepackage{csquotes}
\usepackage[UKenglish]{babel}
\usepackage[UKenglish]{isodate}

\numberwithin{equation}{section}

\newtheorem{theorem}{Theorem}[section]
\newtheorem*{theorem*}{Theorem}
\newtheorem{lemma}[theorem]{Lemma}
\newtheorem{corollary}[theorem]{Corollary}
\newtheorem{defn}[theorem]{Definition}

\newtheorem{prop}[theorem]{Proposition}

\newtheorem*{conjecture*}{Conjecture}
\newtheorem*{idea*}{Main idea}
\theoremstyle{remark}
\newtheorem{remark}[theorem]{Remark}
\theoremstyle{remark}
\newtheorem{example}[theorem]{Example}
\theoremstyle{remark}
\newtheorem*{notation*}{Notation}
\theoremstyle{remark}
\newtheorem*{conventions*}{Conventions}

\newcommand{\R}{\mathbf R}

\newcommand{\Borel}{\mathcal B}
\newcommand{\Fil}{\mathcal F}

\newcommand{\X}{\mathsf X}

\newcommand{\Path}{\mathsf P}
\newcommand{\Prob}{\mathcal P}
\newcommand{\Law}{\mathrm{Law}}

\newcommand{\Range}{\mathrm{Range}}

\newcommand{\dd}{\mathrm d}
\newcommand{\bE}{\mathbf E}
\newcommand{\prob}{\mathbf P}
\newcommand{\Ent}{\mathrm H}

\newcommand{\CMI}{\mathrm I}

\newcommand{\Cov}{\mathrm{Cov}}
\newcommand{\ep}{\mathrm{ep}}
\newcommand{\LL}{\mathscr L}

\newcommand{\one}{\mathbf 1}
\newcommand{\ci}{\perp\!\!\!\perp}

\title{The nonequilibrium statistical mechanics of Markov interacting particles}
\author{Dalton A R Sakthivadivel}
\address{Department of Mathematics, CUNY Graduate Centre, 365 Fifth Avenue, New York, NY 10016}
\email{dsakthivadivel@gc.cuny.edu}
\urladdr{https://darsakthi.github.io}
\date{\today}
\subjclass[2020]{60H07, 60J60, 60J25, 60G15, 60G22, 62F15, 82C31}

\begin{document}

\begin{abstract}
We consider coupled stochastic systems decomposed into exterior, boundary, and interior variables, with the boundary variables sometimes carrying the directed structure of a sensor and actuator.  The central question is when the conditional law of histories factorises, and how this path space statement is detected by log likelihoods, by Girsanov changes of measure, and by information theoretic quantities used in nonequilibrium statistical physics.  The basic object is a regular conditional probability on a path space.  Under domination by clamped reference laws, the boundary property becomes multiplicative separation of a Radon--Nikodym derivative, or equivalently additive separation of a path log likelihood.  For It\=o diffusions this log likelihood is computed by Girsanov's theorem; its expectation is the quadratic control energy appearing in the F\"ollmer entropy identity and in the stochastic control formulation of Schr\"odinger bridge problems.  When exact factorisation fails, the remaining coupling is measured by conditional mutual information, namely the relative entropy between the true boundary-conditioned path law and the product of its conditional marginals.  This gives a common language for boundary screening, path likelihood inference, controlled changes of path law, and the thermodynamic value of mutual information.
\end{abstract}

\maketitle

\setcounter{tocdepth}{1}

\tableofcontents

\section{Introduction}

A Markov boundary is often introduced as a graphical separation condition.  In the present discussion the graph is not the object of study.  The object of study is the conditional law of paths.  One begins with a stochastic process
\[
	Z_t = (Y_t,B_t,X_t),
\]
where $Y$ denotes exterior variables, $X$ denotes interior variables, and $B$ denotes boundary variables.  The instantaneous boundary statement is
\[
	Y_t \ci X_t \mid B_t.
\]
The path space boundary statement is
\[
	Y_{[0,T]} \ci X_{[0,T]} \mid B_{[0,T]}.
\]
These statements are not equivalent.  The second is a statement about a regular conditional probability on a Polish path space.  It rules out dependence through hidden innovations, memory variables, temporal feedback channels not screened off by the boundary history, and other forms of collider effect which generate latent higher order dependency.  The first only rules out instantaneous conditional dependence after the current boundary value has been observed.

The probabilistic formulation is close to the mathematical literature on path likelihoods and measures on function spaces.  The Onsager--Machlup functional began as a path density for fluctuations of irreversible processes, and later became a precise way to define modes, maximum a posteriori estimators, and small ball asymptotics when no Lebesgue reference measure exists on the ambient function space \cite{OnsagerMachlup1953,DemboZeitouni1991,Stuart2010,DashtiLawStuartVoss2013,AyanbayevKlebanovLieSullivan2021}.  Path integral formulations of inverse problems express the same idea in the language of actions and partition functions \cite{ChangSavageChou2014}.  Our use of likelihood is more modest.  The Radon--Nikodym derivative of a boundary-conditioned path law with respect to clamped product reference laws is the object of interest, and the question is whether its logarithm contains an irreducible mixed exterior-interior term.

The control interpretation comes from Girsanov's theorem.  When two diffusion path laws have the same diffusion coefficient and different drifts, their likelihood ratio is an exponential martingale; taking expectation under the controlled law leaves exactly the quadratic drift energy \cite{Girsanov1960}.  F\"ollmer's entropy method, the Bou\'e--Dupuis representation, and the modern Schr\"odinger bridge literature make this into the identity that relative entropy on path space is the least adapted control energy needed to realise the target law from the reference Brownian law \cite{Foellmer1985,BoueDupuis1998,Lehec2013,Leonard2014,ChenGeorgiouPavon2016}.  The boundary problem studied here is obtained by applying this identity after a boundary path has been fixed.

The information theoretic interpretation is equally classical in statistical physics.  Mutual information is a relative entropy between a joint law and the product of its marginals; conditional mutual information is the same comparison performed fibrewise over the boundary.  Stochastic thermodynamics uses mutual information to quantify information flow, prediction, and the thermodynamic effect of statistical dependence \cite{StillSivakBellCrooks2012,HorowitzEsposito2014,ParrondoHorowitzSagawa2015}.  Recent work in complex systems has emphasised that mutual information can separate intrinsic interactions from environmental or multiscale contributions, and can detect how information propagates across timescales \cite{NicolettiBusiello2021,NicolettiBusiello2022,NicolettiBusiello2024}.  This motivates treating the boundary error as a conditional mutual information on path space rather than as a coordinate dependent norm.

The resulting picture is as follows.  A boundary is a disintegration kernel which factorises.  Under domination, this is a log likelihood separation criterion.  In It\=o models, Girsanov computes the log likelihood and identifies its expectation with a control energy.  In nonequilibrium statistical physics, the same relative entropy is read as the information lost by imposing a separated ensemble, and, after multiplying by the thermal scale, as an excess free energy or work lower bound under the usual thermodynamic conventions.

\section{Preliminaries on measure theory}

\subsection{Standard Borel spaces and kernels}

\begin{defn}
Let $(X,\Borel_X)$ and $(Y,\Borel_Y)$ be measurable spaces.  A Markov kernel from $Y$ to $X$ is a map
\[
        K\colon Y\times\Borel_X\longrightarrow [0,1]
\]
such that $K(y,\cdot)$ is a probability measure on $X$ for every $y$, and $y\mapsto K(y,A)$ is $\Borel_Y$-measurable for every $A\in\Borel_X$.  We write $K_y(\dd x)$ or $K(\dd x\mid y)$.
\end{defn}

\begin{defn}
Let $X,Y$ be standard Borel spaces and let $P\in\Prob(X\times Y)$.  A regular conditional probability of $X$ given $Y$ is a kernel $K:Y\to\Prob(X)$ such that
\[
        P(A\times C)=\int_C K_y(A)\,P_Y(\dd y)
\]
for all $A\in\Borel_X$ and $C\in\Borel_Y$.  We write
\[
        P(\dd x,\dd y)=P(\dd x\mid y)P_Y(\dd y).
\]
\end{defn}

The following theorem is standard and establishes the result necessary concerning the existence and essential uniqueness of disintegration. We will give the argument for completeness.

\begin{theorem}
Let $X,Y$ be standard Borel spaces and let $P\in\Prob(X\times Y)$.  Then there exists a regular conditional probability $P(\dd x\mid y)$.  If $K$ and $K'$ are two such kernels, then
\[
        K_y=K'_y
\]
for $P_Y$-almost every $y$.
\end{theorem}

\begin{proof}
This is the standard disintegration theorem for standard Borel spaces.  The existence follows by identifying standard Borel spaces with Borel subsets of $[0,1]$, constructing conditional distribution functions on a countable generating class, and extending by monotone class.  Essential uniqueness follows from the same countable generator.  If $\mathcal C$ is a countable $\pi$-system generating $\Borel_X$, then for each $A\in\mathcal C$ the functions $y\mapsto K_y(A)$ and $y\mapsto K'_y(A)$ are two versions of $\bE[\one_{\{X\in A\}}\mid Y=y]$, hence agree $P_Y$-a.s.  Intersecting the exceptional null sets over the countable class $\mathcal C$, the two kernels agree on $\mathcal C$ for all remaining $y$.  The monotone class theorem then gives equality on all Borel sets.
\end{proof}

\subsection{Conditional independence}

\begin{defn}
Let $X,Y,Z$ be random elements in standard Borel spaces.  We say that $X$ and $Y$ are conditionally independent given $Z$, and write
\[
        X\ci Y\mid Z,
\]
if for every bounded measurable $f$ and $g$,
\[
        \bE[f(X)g(Y)\mid Z]
        =\bE[f(X)\mid Z]\,\bE[g(Y)\mid Z]
\]
almost surely.
\end{defn}

The following lemma gives the kernel factorisation criterion.

\begin{lemma}
Let $X,Y,Z$ be standard Borel random elements with joint law $P$.  Let
\[
        P(\dd x,\dd y,\dd z)=P_{X,Y\mid Z=z}(\dd x,\dd y)P_Z(\dd z)
\]
be a disintegration.  Then
\[
        X\ci Y\mid Z
\]
if and only if
\begin{equation}\label{kernel-factorisation-eq}
        P_{X,Y\mid Z=z}
        =P_{X\mid Z=z}\otimes P_{Y\mid Z=z}
\end{equation}
for $P_Z$-almost every $z$.
\end{lemma}

\begin{proof}
Assume conditional independence.  For bounded measurable $f,g$,
\[
        \int f(x)g(y)\,P_{X,Y\mid Z=z}(\dd x,\dd y)
        =\int f\,\dd P_{X\mid Z=z}\int g\,\dd P_{Y\mid Z=z}
\]
for $P_Z$-a.e. $z$.  By taking $f=\one_A$ and $g=\one_C$ on a countable generating class and applying a monotone class argument, the conditional measure agrees with the product measure on all rectangles, hence on the product sigma-field.  This proves \eqref{kernel-factorisation-eq}.  Conversely, if \eqref{kernel-factorisation-eq} holds, then integrating $f(x)g(y)$ against the factorised kernel gives the defining identity of conditional independence.
\end{proof}

The following lemma discusses Radon--Nikodym separability.

\begin{lemma}
Let $(X,\mathcal X)$ and $(Y,\mathcal Y)$ be standard Borel spaces.  Let $\lambda\in\Prob(X\times Y)$ and suppose
\[
        \lambda\ll \alpha\otimes\beta
\]
for probability measures $\alpha\in\Prob(X)$, $\beta\in\Prob(Y)$.  Write
\[
        h(x,y)=\frac{\dd\lambda}{\dd(\alpha\otimes\beta)}(x,y).
\]
Then $\lambda=\lambda_X\otimes\lambda_Y$ if and only if there exist non-negative measurable functions $u:X\to[0,\infty)$ and $v:Y\to[0,\infty)$ such that
\begin{equation}\label{rn-sep-eq}
        h(x,y)=u(x)v(y)
\end{equation}
for $\alpha\otimes\beta$-almost every $(x,y)$.
\end{lemma}

\begin{proof}
If $h=uv$, then
\[
        \lambda(\dd x,\dd y)=u(x)\alpha(\dd x)v(y)\beta(\dd y).
\]
Since $\lambda$ is a probability measure, the product of the two total masses is one.  Its marginals are
\[
        \lambda_X(\dd x)=c_Yu(x)\alpha(\dd x), 
        \lambda_Y(\dd y)=c_Xv(y)\beta(\dd y),
\]
where $c_X=\int u\,\dd\alpha$ and $c_Y=\int v\,\dd\beta$.  Since $c_Xc_Y=1$, $\lambda=\lambda_X\otimes\lambda_Y$.

Conversely, if $\lambda=\lambda_X\otimes\lambda_Y$ and $\lambda\ll\alpha\otimes\beta$, then $\lambda_X\ll\alpha$ and $\lambda_Y\ll\beta$.  Indeed, if $\alpha(A)=0$, then $(\alpha\otimes\beta)(A\times Y)=0$, hence $\lambda_X(A)=\lambda(A\times Y)=0$.  Similarly for $Y$.  Thus
\[
        u=\frac{\dd\lambda_X}{\dd\alpha},
        v=\frac{\dd\lambda_Y}{\dd\beta}
\]
exist and
\[
        \frac{\dd\lambda}{\dd(\alpha\otimes\beta)}(x,y)=u(x)v(y)
\]
by uniqueness of Radon--Nikodym derivatives.
\end{proof}


\begin{lemma}
Let $X,Y,Z$ be standard Borel random elements.  Assume regular conditional probabilities have been chosen.  Then
\begin{equation}\label{cmi-rn-eq}
        \CMI(X;Y\mid Z)
        =\int \Ent\bigl(P_{X,Y\mid Z=z}\mid P_{X\mid Z=z}\otimes P_{Y\mid Z=z}\bigr)\,P_Z(\dd z),
\end{equation}
with the convention that the integrand is $+\infty$ if absolute continuity fails.  In particular, $X\ci Y\mid Z$ if and only if the right hand side vanishes.
\end{lemma}

\begin{proof}
The conditional mutual information is defined by the relative entropy
\[
        \Ent(P_{X,Y,Z}\mid P_{X\mid Z}\otimes P_{Y\mid Z}\otimes P_Z),
\]
where the second measure denotes the measure obtained by integrating the product conditional kernel over $P_Z$.  The chain rule for relative entropy under disintegration gives \eqref{cmi-rn-eq}.  If the integral vanishes, the integrand vanishes for $P_Z$-a.e. $z$.  Relative entropy vanishes exactly when the two probability measures coincide.  The kernel factorisation criterion then gives conditional independence.  The converse is immediate.
\end{proof}

\section{Path spaces and conditioning of trajectories}

\begin{defn}
Let $S$ be a Polish space.  The continuous path space over $[0,T]$ is
\[
        \Path_T(S)=C([0,T];S)
\]
with the topology of uniform convergence.  If $Z=(Z_t)_{0\leqslant t\leqslant T}$ is a continuous $S$-valued process, its path law is
\[
        \Law(Z_{[0,T]})\in\Prob(\Path_T(S)).
\]
\end{defn}

Since $S$ is Polish, $\Path_T(S)$ is Polish.  Hence all disintegration results above apply to path laws.

Conditioning one trajectory on another satisfies the following.

\begin{prop}
Let $X$ and $Y$ be continuous processes with Polish state spaces.  Put
\[
        Q=\Law(X_{[0,T]},Y_{[0,T]})
        \in\Prob(\Path_T(S_X)\times\Path_T(S_Y)).
\]
Then there exists a kernel
\[
        y\longmapsto Q^y\in\Prob(\Path_T(S_X))
\]
such that
\[
        Q(\dd x,\dd y)=Q^y(\dd x)Q_Y(\dd y).
\]
The conditional law $Q^y$ is the mathematically precise meaning of
\[
        \Law(X_{[0,T]}\mid Y_{[0,T]}=y).
\]
It is unique for $Q_Y$-almost every $y$.
\end{prop}

\begin{proof}
Apply the disintegration theorem to the standard Borel spaces $\Path_T(S_X)$ and $\Path_T(S_Y)$.
\end{proof}


\begin{prop}
Let $(X_t)_{t\geqslant 0}$ be a time-homogeneous Markov process on a Polish space $S$, and let $P^x$ be its law on the canonical path space starting at $x$.  Fix $s<T$.  Write $\omega^- = \omega|_{[0,s]}$ and $\omega^+=\omega|_{[s,T]}$.  Then
\[
        P^x(\dd\omega^-,\dd\omega^+)
        =P^x(\dd\omega^-)K_s(\omega^-,\dd\omega^+),
\]
where
\[
        K_s(\omega^-,\cdot)=P^{\omega^-(s)}(\cdot)
\]
up to the harmless convention that the future path starts at time $s$.
\end{prop}

\begin{proof}
For bounded measurable functionals $F$ of the past and $G$ of the future,
\[
        \bE^x[F(X_{[0,s]})G(X_{[s,T]})]
        =\bE^x\left[F(X_{[0,s]})\bE^x[G(X_{[s,T]})\mid\Fil_s]\right].
\]
The Markov property gives
\[
        \bE^x[G(X_{[s,T]})\mid\Fil_s]
        =\bE^{X_s}[G(X_{[0,T-s]})],
\]
with time shifted notation.  This is exactly the stated kernel identity.  Uniqueness follows from the uniqueness of regular conditional probabilities.
\end{proof}

\begin{remark}
Conditioning on the full path of another process usually destroys Markovianity.  Even if $(X,Y)$ is jointly Markov, the smoothing law $\Law(X_{[0,T]}\mid Y_{[0,T]}=y)$ is typically a Gibbs factor on the prior path law and depends on the whole observed path $y$, including its future relative to intermediate times.
\end{remark}

\section{Preliminaries on Gaussian processes and abstract Wiener spaces}

\subsection{Gaussian disintegration}


\begin{defn}
An abstract Wiener space is a triple $(H,B,\gamma)$, where $B$ is a separable Banach space, $H\subset B$ is a continuously and densely embedded Hilbert space, and $\gamma$ is a centred Gaussian measure on $B$ whose Cameron--Martin space is $H$.
\end{defn}

Since $B$ is Polish, all regular conditional probabilities exist.  The additional Gaussian structure gives explicit conditional laws for linear observations.

The following theorem establishes Gaussian conditioning by a linear observation.

\begin{theorem}
Let $(X,Y)$ be a centred Gaussian random element in a product of separable Hilbert spaces $H_X\times H_Y$.  Assume the covariance operators are
\[
        \Sigma=\begin{pmatrix}
            \Sigma_{XX} & \Sigma_{XY}\\
            \Sigma_{YX} & \Sigma_{YY}
        \end{pmatrix},
\]
with $\Sigma_{YY}$ injective on the closed support of $Y$.  Then a regular conditional law of $X$ given $Y=y$ is Gaussian with mean
\[
        m_{X\mid y}=\Sigma_{XY}\Sigma_{YY}^{\dagger}y
\]
and covariance
\[
        \Sigma_{X\mid Y}
        =\Sigma_{XX}-\Sigma_{XY}\Sigma_{YY}^{\dagger}\Sigma_{YX},
\]
where $\dagger$ denotes the Moore--Penrose inverse on the support.  The same formula holds in finite-dimensional projections and passes to the Hilbert limit whenever the displayed operators define a Gaussian covariance.
\end{theorem}

\begin{proof}
For finite-dimensional Gaussian vectors this is the Schur-complement formula, obtained by completing the square in the joint Gaussian density.  Equivalently, choose the linear least-squares estimator $AY$ of $X$, where $A=\Sigma_{XY}\Sigma_{YY}^{\dagger}$.  Then
\[
        \Cov(X-AY,Y)=\Sigma_{XY}-\Sigma_{XY}\Sigma_{YY}^{\dagger}\Sigma_{YY}=0
\]
on the support of $Y$.  For Gaussian variables, uncorrelatedness implies independence.  Therefore
\[
        X=AY+\xi,
\]
where $\xi$ is independent of $Y$ and has covariance $\Sigma_{XX}-\Sigma_{XY}\Sigma_{YY}^{\dagger}\Sigma_{YX}$.  Conditioning on $Y=y$ gives the stated Gaussian law.  In the Hilbert case apply the finite-dimensional result to cylindrical projections and use consistency of Gaussian finite-dimensional distributions.
\end{proof}

The following theorem establishes Gaussian conditional independence.

\begin{theorem}
Let $(Y,B,X)$ be a centred Gaussian random element in a finite-dimensional Euclidean space, or more generally in a Hilbert space with non-degenerate covariance on the relevant support.  Let $K=\Sigma^{-1}$ be the precision operator in block form.  Then the following are equivalent:
\begin{enumerate}
\item $Y\ci X\mid B$.
\item The conditional covariance block vanishes:
\[
        \Sigma_{YX\mid B}:=\Sigma_{YX}-\Sigma_{YB}\Sigma_{BB}^{\dagger}\Sigma_{BX}=0.
\]
\item The $(Y,X)$-block of the precision operator vanishes:
\[
        K_{YX}=0.
\]
\end{enumerate}
\end{theorem}

\begin{proof}
By the Gaussian conditioning theorem, $(Y,X)\mid B=b$ is Gaussian with covariance
\[
        \begin{pmatrix}
        \Sigma_{YY\mid B} & \Sigma_{YX\mid B}\\
        \Sigma_{XY\mid B} & \Sigma_{XX\mid B}
        \end{pmatrix}.
\]
A Gaussian vector has independent components if and only if their covariance block is zero.  Thus (1) and (2) are equivalent.  For (2) and (3), invert the block covariance matrix with the $B$-block separated.  The Schur complement of $\Sigma_{BB}$ in $\Sigma$ is the conditional covariance of $(Y,X)$ given $B$.  The inverse of this Schur complement is the $(Y,X)$-principal block of the full precision after eliminating $B$.  In particular, the off-diagonal block of the conditional precision is zero if and only if the off-diagonal block of the conditional covariance is zero.  This off-diagonal block agrees with $K_{YX}$ in the full precision block decomposition.  Equivalently, write the Gaussian density as
\[
        p(y,b,x)\propto
        \exp\left[-\frac12
        \begin{pmatrix}y\\ b\\ x\end{pmatrix}^{\!\top}
        K
        \begin{pmatrix}y\\ b\\ x\end{pmatrix}
        \right].
\]
For fixed $b$, the only term coupling $y$ and $x$ is $y^\top K_{YX}x$.  The conditional density factorises in $(y,x)$ if and only if this term vanishes.  This proves the equivalence.
\end{proof}

\begin{remark}
The last density proof is the most direct way to remember the result.  For Gaussian measures, conditional independence is a statement about zeros in the precision operator, not merely zeros in the covariance operator.  This is especially important on path space, where covariance kernels encode temporal memory and precision kernels encode conditional local constraints.
\end{remark}

\section{Path space boundaries}

\begin{defn}
Let
\[
        Z_t=(Y_t,B_t,X_t)
\]
be a continuous process with state space
\[
        \X=\X_Y\times\X_B\times\X_X.
\]
Its joint path law over $[0,T]$ is
\[
        P_T=\Law(Y_{[0,T]},B_{[0,T]},X_{[0,T]})
        \in \Prob(\Path_T(\X_Y)\times\Path_T(\X_B)\times\Path_T(\X_X)).
\]
\end{defn}

The definition to follow introduces the notion of path space Markov boundary.

\begin{defn}
The boundary path $B_{[0,T]}$ is a path space Markov boundary between $Y_{[0,T]}$ and $X_{[0,T]}$ if
\[
        Y_{[0,T]}\ci X_{[0,T]}\mid B_{[0,T]}.
\]
Equivalently, for $P_B$-almost every $b$,
\begin{equation}\label{pathboundary-eq}
        P_T(\dd y,\dd x\mid b)
        =P_T(\dd y\mid b)P_T(\dd x\mid b).
\end{equation}
\end{defn}

Let $b$ be a boundary path for which the conditional law $P^b$ is dominated by a reference measure $Q^b$ on the hidden path space.  The path log likelihood of $P^b$ relative to $Q^b$ is any measurable representative
\[
	\ell_b(y,x) = \log \frac{\dd P^b}{\dd Q^b}(y,x),
\]
with value $-\infty$ on the set where the derivative vanishes.  When $Q^b = Q_Y^b \otimes Q_X^b$, the absence of a mixed exterior-interior likelihood term means that there are measurable functions $a_b$, $c_b$, and $k(b)$ such that $\ell_b(y,x) = a_b(y) + c_b(x) + k(b)$ outside a null set.  This is the form in which the path space likelihood criterion will be used below.

The following theorem is one of our main results on the necessary and sufficient conditions by which a path space law factorises.

\begin{theorem}
Let $P_T$ be the joint path law of $(Y,B,X)$.  Choose a disintegration
\[
        P_T(\dd y,\dd b,\dd x)=P^b(\dd y,\dd x)P_B(\dd b).
\]
For each $b$, suppose $P^b$ is dominated by a product reference measure
\[
        Q_Y^b\otimes Q_X^b
\]
on $\Path_T(\X_Y)\times\Path_T(\X_X)$, and write
\[
        H_b(y,x)=\frac{\dd P^b}{\dd(Q_Y^b\otimes Q_X^b)}(y,x).
\]
Then $B_{[0,T]}$ is a path space Markov boundary if and only if for $P_B$-almost every $b$ there exist non-negative measurable functionals $U_b$ and $V_b$ such that
\begin{equation}\label{path-rn-factorisation-eq}
        H_b(y,x)=U_b(y)V_b(x)
\end{equation}
for $Q_Y^b\otimes Q_X^b$-almost every $(y,x)$.
\end{theorem}

\begin{proof}
For a fixed $b$, apply the Radon--Nikodym separability lemma to the conditional measure $P^b$.  It says that $P^b$ is a product measure if and only if its density with respect to the product reference $Q_Y^b\otimes Q_X^b$ is multiplicatively separable.  The kernel factorisation criterion says that the path space boundary property is precisely the assertion that $P^b$ is a product measure for $P_B$-almost every $b$.  Combining the two statements gives the theorem.
\end{proof}

\begin{corollary}
Assume that, after fixing a boundary path $b$, there are prior or clamped path measures $Q_Y^b,Q_X^b$ and a boundary likelihood $L_B(b\mid y,x)$ such that
\begin{equation}\label{bayes-path-eq}
        P^b(\dd y,\dd x)
        =\frac{L_B(b\mid y,x)}{Z(b)}Q_Y^b(\dd y)Q_X^b(\dd x).
\end{equation}
Then $B_{[0,T]}$ is a path space Markov boundary if and only if for $P_B$-almost every $b$,
\[
        L_B(b\mid y,x)=U_b(y)V_b(x)
\]
up to a $(Q_Y^b\otimes Q_X^b)$-null set.  Equivalently,
\[
        \log L_B(b\mid y,x)=a_b(y)+c_b(x)+k(b)
\]
whenever the likelihood is strictly positive.
\end{corollary}

\begin{proof}
In \eqref{bayes-path-eq}, the Radon--Nikodym density of $P^b$ with respect to $Q_Y^b\otimes Q_X^b$ is $H_b=L_B/Z(b)$.  Multiplication by the normalising constant $Z(b)^{-1}$ does not affect multiplicative separability.  If the likelihood is strictly positive, take logarithms.  Multiplicative separation becomes additive separation plus a $b$-dependent constant.
\end{proof}

\begin{remark}
This theorem is the most direct form of the necessary and sufficient condition.  All later stochastic differential equation, Gaussian, and random dynamical system criteria are ways of computing or structurally forcing the likelihood or Radon--Nikodym density in \eqref{path-rn-factorisation-eq}.
\end{remark}

\section{It\=o systems with Brownian innovations}

\subsection{The joint It\=o equation and clamped laws}

Consider an It\=o diffusion on $\R^{d_Y}\times\R^{d_B}\times\R^{d_X}$:
\begin{equation}\label{joint-ito-eq}
\begin{aligned}
        \dd Y_t&=b_Y(Y_t,B_t,X_t)\,\dd t+
                 \sigma_Y(Y_t,B_t,X_t)\,\dd W^Y_t,\\
        \dd B_t&=b_B(Y_t,B_t,X_t)\,\dd t+
                 \sigma_B(Y_t,B_t,X_t)\,\dd W^B_t,\\
        \dd X_t&=b_X(Y_t,B_t,X_t)\,\dd t+
                 \sigma_X(Y_t,B_t,X_t)\,\dd W^X_t.
\end{aligned}
\end{equation}
The Brownian motions may have several components.  The phrase \emph{Brownian innovation boundary model} will mean that the driving Brownian innovations can be chosen so that, after conditioning on the boundary path, the exterior and interior innovation noises are independent unless explicitly coupled through the coefficients.

The most useful formulation fixes a boundary path $b$.  One then compares the conditional law of $(Y,X)$ given $B=b$ with a product of clamped laws.  The clamped exterior law $Q_Y^b$ is the law of the exterior equation with the boundary path inserted and with interior feedback removed or held at the reference structure.  Similarly $Q_X^b$ is the clamped interior law.  The exact construction depends on the model, but the theorem only needs domination by such product references.

The following theorem is another of our main results concerning factorisation of laws of the paths of an It\=o process.

\begin{theorem}
Assume the joint It\=o system \eqref{joint-ito-eq} is weakly well posed on $[0,T]$, and assume that for $P_B$-almost every boundary path $b$ the conditional law
\[
        \Pi^b:=\Law(Y_{[0,T]},X_{[0,T]}\mid B_{[0,T]}=b)
\]
is dominated by a clamped product law $Q_Y^b\otimes Q_X^b$.  Suppose further that the conditional density has the Bayes--Girsanov form
\begin{equation}\label{girsanov-bayes-density-eq}
        \frac{\dd\Pi^b}{\dd(Q_Y^b\otimes Q_X^b)}(y,x)
        =\frac{L_B(b\mid y,x)}{Z(b)}.
\end{equation}
Then
\[
        \Law(Y_{[0,T]},X_{[0,T]}\mid B_{[0,T]}=b)
        =\Law(Y_{[0,T]}\mid B_{[0,T]}=b)
         \otimes
         \Law(X_{[0,T]}\mid B_{[0,T]}=b)
\]
for $P_B$-almost every $b$ if and only if
\begin{equation}\label{white-likelihood-separates-eq}
        L_B(b\mid y,x)=U_b(y)V_b(x)
\end{equation}
for $Q_Y^b\otimes Q_X^b$-almost every $(y,x)$.  If the likelihood is positive, this is equivalent to
\begin{equation}\label{white-loglik-separates-eq}
        \log L_B(b\mid y,x)=A_b(y)+C_b(x)+K(b).
\end{equation}
\end{theorem}

\begin{proof}
This is the Bayes form of the path space factorisation theorem applied to $\Pi^b$.  The role of the It\=o and Brownian innovation assumptions is to justify the existence of the clamped references and the Girsanov density \eqref{girsanov-bayes-density-eq}.  Once \eqref{girsanov-bayes-density-eq} holds, the conditional path law factorises if and only if its density with respect to the clamped product law factorises.  This density is $L_B/Z(b)$, and the normalising constant does not affect separability.  Positivity allows logarithms, converting multiplicative separability into additive separability.
\end{proof}

\subsection{Girsanov form}

To see the content of \eqref{white-loglik-separates-eq}, take the boundary equation in observation form
\begin{equation}\label{boundary-observation-eq}
        \dd B_t=h(Y_t,X_t,B_t)\,\dd t+\Sigma_B(B_t)\,\dd W^B_t,
\end{equation}
with uniformly non-degenerate $a_B=\Sigma_B\Sigma_B^\top$.  Relative to a reference boundary process with drift $h_0(B_t)$, the likelihood of observing $b$ given $(y,x)$ is formally
\begin{equation}\label{girsanov-loglik-eq}
\begin{aligned}
        \log L_B(b\mid y,x)
        &=\int_0^T
        \langle \Sigma_B^{-1}(h(y_t,x_t,b_t)-h_0(b_t)),\dd W^{0,b}_t\rangle\\
        & -\frac12\int_0^T
        \|\Sigma_B^{-1}(h(y_t,x_t,b_t)-h_0(b_t))\|^2\,\dd t,
\end{aligned}
\end{equation}
where $W^{0,b}$ is the innovation path reconstructed under the reference model.

The following proposition demonstrates that there is no mixed likelihood term under factorisation.

\begin{prop}
Assume the observation model \eqref{boundary-observation-eq} is valid and that Novikov's condition holds.  If, for the fixed boundary path $b$, the drift difference decomposes as
\[
        \Sigma_B^{-1}(h(y_t,x_t,b_t)-h_0(b_t))
        =r_Y(y_t,b_t)+r_X(x_t,b_t)
\]
then the likelihood separates if and only if the mixed energy term
\begin{equation}\label{mixed-energy-term-eq}
        \int_0^T \langle r_Y(y_t,b_t),r_X(x_t,b_t)\rangle\,\dd t
\end{equation}
separates additively as a functional of $y$ plus a functional of $x$, up to a constant depending on $b$.  In particular, if
\[
        \langle r_Y(y,b),r_X(x,b)\rangle=0
\]
for all $(y,x,b)$, then the likelihood separates.
\end{prop}

\begin{proof}
Insert the decomposition into \eqref{girsanov-loglik-eq}.  The stochastic integral is the sum of an $(y,b)$-functional and an $(x,b)$-functional.  The quadratic energy is
\[
        \frac12\int_0^T\|r_Y\|^2\dd t
        +\frac12\int_0^T\|r_X\|^2\dd t
        +\int_0^T\langle r_Y,r_X\rangle\dd t.
\]
The first two terms already separate.  Therefore the whole log-likelihood separates if and only if the mixed term separates.  Orthogonality makes the mixed term vanish, which is a sufficient condition.
\end{proof}

The same computation gives the control identity behind the use of Girsanov likelihoods.

\begin{theorem}\label{girsanov-control-identity-thm}
Let $R$ be Wiener measure on $C([0,T];\mathbf{R}^d)$ started at zero, and let $P \ll R$ be a probability measure of finite entropy.  Suppose the canonical process has the decomposition
\[
	X_t = \int_0^t u_s\,\dd s + W_t^P
\]
under $P$, where $W^P$ is a $P$-Brownian motion and $u$ is progressively measurable.  Assume
\[
	\bE_P\int_0^T \|u_t\|^2\,\dd t < \infty.
\]
Then
\[
	\Ent(P \mid R) = \frac12\bE_P\int_0^T \|u_t\|^2\,\dd t.
\]
Equivalently, the expected log likelihood of the controlled law relative to the uncontrolled law is the quadratic control energy.
\end{theorem}

\begin{proof}
Girsanov's theorem gives the density
\[
	\frac{\dd P}{\dd R}
	= \exp\left(\int_0^T \langle u_t,\dd X_t\rangle - \frac12\int_0^T \|u_t\|^2\,\dd t\right),
\]
where the stochastic integral is read in the usual exponential martingale sense by localisation.  Under $P$ one has $\dd X_t = u_t\,\dd t + \dd W_t^P$, and therefore
\[
	\log \frac{\dd P}{\dd R}
	= \frac12\int_0^T \|u_t\|^2\,\dd t + \int_0^T \langle u_t,\dd W_t^P\rangle.
\]
The stochastic integral has mean zero after localisation and passage to the limit by the finite energy assumption.  Taking $P$-expectations gives the identity.  In F\"ollmer's formulation the drift $u$ is the entropy-optimal adapted drift realising $P$ from the Wiener reference; thus the equality also expresses the least control energy representation of path space relative entropy.
\end{proof}


\begin{theorem}
Let $b\in\Path_T(\X_B)$ be fixed.  Suppose the conditional law $\Pi^b$ of $(Y,X)$ is Markov with time-inhomogeneous generator $\LL_t^b$ on $\X_Y\times\X_X$, defined on smooth compactly supported functions, with local form
\[
\begin{aligned}
        \LL_t^b f(y,x)
        &=\sum_a \beta_Y^a(t,y,x;b)\partial_{y_a}f
          +\sum_r \beta_X^r(t,y,x;b)\partial_{x_r}f\\
        & +\frac12\sum_{a,c} A_{YY}^{ac}(t,y,x;b)\partial_{y_a y_c}^2f
          +\sum_{a,r} A_{YX}^{ar}(t,y,x;b)\partial_{y_a x_r}^2f\\
        & +\frac12\sum_{r,s} A_{XX}^{rs}(t,y,x;b)\partial_{x_r x_s}^2f.
\end{aligned}
\]
Assume uniqueness of the martingale problem.  Then $\Pi^b$ is a product law $\Pi_Y^b\otimes\Pi_X^b$ with independent coordinate processes if and only if, after modification on null sets,
\begin{equation}\label{generator-splitting-eq}
        \LL_t^b=\LL_{Y,t}^b\otimes 1+1\otimes\LL_{X,t}^b,
\end{equation}
that is,
\[
        A_{YX}=0, 
        \beta_Y(t,y,x;b)=\beta_Y(t,y;b), 
        A_{YY}(t,y,x;b)=A_{YY}(t,y;b),
\]
and
\[
        \beta_X(t,y,x;b)=\beta_X(t,x;b), 
        A_{XX}(t,y,x;b)=A_{XX}(t,x;b).
\]
\end{theorem}

\begin{proof}
If $\Pi^b=\Pi_Y^b\otimes\Pi_X^b$ and the coordinate processes are Markov with generators $\LL_{Y,t}^b$ and $\LL_{X,t}^b$, then for a product test function $f(y,x)=u(y)v(x)$, the product semigroup satisfies
\[
        \frac{\dd}{\dd t}\bE[u(Y_t)v(X_t)]
        =\bE[(\LL_{Y,t}^b u)(Y_t)v(X_t)+u(Y_t)(\LL_{X,t}^b v)(X_t)].
\]
Thus its generator is \eqref{generator-splitting-eq}.  Comparing the local second-order expression with this split generator gives no mixed second derivatives and no coefficient dependence on the opposite variable.

Conversely, assume \eqref{generator-splitting-eq}.  Let $P_Y^b$ and $P_X^b$ be the laws solving the two marginal martingale problems with the corresponding initial marginals.  The product law $P_Y^b\otimes P_X^b$ solves the martingale problem for $\LL_t^b$, because for product test functions It\=o's formula gives the sum of the two marginal martingales, and a monotone-class/core argument extends this to the full domain.  By uniqueness of the martingale problem, $\Pi^b=P_Y^b\otimes P_X^b$.  Hence the conditional path law factorises.
\end{proof}

\section{Sensory-active asymmetry}

The boundary often splits as
\[
        B=(S,A),
\]
where $S$ are sensory states influenced by the exterior and observed by the interior, while $A$ are active states influenced by the interior and acting on the exterior.  The structural picture is
\[
        Y \longrightarrow S \longrightarrow X \longrightarrow A \longrightarrow Y.
\]
This describes directed dynamical separation, not a spatially symmetric Markov random field.


\begin{defn}
An It\=o sensor-actuator system driven by Brownian innovations is a system of the form
\begin{equation}\label{sensor-actuator-system-eq}
\begin{aligned}
        \dd Y_t&=b_Y(Y_t,A_t)\,\dd t+\sigma_Y(Y_t,A_t)\,\dd W^Y_t,\\
        \dd S_t&=b_S(Y_t,S_t,A_t)\,\dd t+\sigma_S(S_t,A_t)\,\dd W^S_t,\\
        \dd X_t&=b_X(X_t,S_t)\,\dd t+\sigma_X(X_t,S_t)\,\dd W^X_t,\\
        \dd A_t&=b_A(X_t,S_t,A_t)\,\dd t+\sigma_A(S_t,A_t)\,\dd W^A_t,
\end{aligned}
\end{equation}
where the Brownian blocks $W^Y,W^S,W^X,W^A$ are independent and the equations are weakly well posed.
\end{defn}

Another main result concerns factorisation of path laws of sensor-actuator systems.

\begin{theorem}
Assume \eqref{sensor-actuator-system-eq}.  Suppose the initial law satisfies
\[
        \Law(Y_0,X_0\mid S_0,A_0)
        =\Law(Y_0\mid S_0,A_0)\otimes\Law(X_0\mid S_0,A_0).
\]
Fix boundary paths $(s,a)$.  Let $Q_Y^a$ be the law of the exterior equation clamped by $a$, and let $Q_X^s$ be the law of the interior equation clamped by $s$.  Assume the sensory and active observation likelihoods exist and are given by $L_S(s\mid y,a)$ and $L_A(a\mid x,s)$.  Then, for $P_{S,A}$-almost every $(s,a)$,
\begin{equation}\label{sensor-active-factorisation-eq}
\begin{aligned}
&\Law(Y_{[0,T]},X_{[0,T]}\mid S_{[0,T]}=s,A_{[0,T]}=a)\\
&  =\Pi_Y^{s,a}\otimes\Pi_X^{s,a},
\end{aligned}
\end{equation}
where
\[
        \Pi_Y^{s,a}(\dd y)\propto L_S(s\mid y,a)Q_Y^a(\dd y),
        \Pi_X^{s,a}(\dd x)\propto L_A(a\mid x,s)Q_X^s(\dd x).
\]
\end{theorem}

\begin{proof}
Under the clamped boundary paths $(s,a)$, the exterior equation depends on $(y,a)$ and its own Brownian motion $W^Y$, while the interior equation depends on $(x,s)$ and its own Brownian motion $W^X$.  The Brownian blocks are independent, and the initial conditional law factorises.  Hence the clamped prior law of $(Y,X)$ is
\[
        Q_Y^a\otimes Q_X^s.
\]
The law of the observed sensory path $s$ depends on the exterior path and the clamped active path through the likelihood $L_S(s\mid y,a)$.  The law of the observed active path $a$ depends on the interior path and the clamped sensory path through $L_A(a\mid x,s)$.  Therefore Bayes' formula gives
\[
        \Pi^{s,a}(\dd y,\dd x)
        \propto L_S(s\mid y,a)L_A(a\mid x,s)
        Q_Y^a(\dd y)Q_X^s(\dd x).
\]
The density is the product of an $y$-functional and an $x$-functional once $(s,a)$ is fixed.  The Radon--Nikodym separability theorem therefore implies the product decomposition \eqref{sensor-active-factorisation-eq}, with the displayed marginal conditional laws.
\end{proof}

\begin{remark}
The asymmetry changes the interpretation but not the disintegration logic.  The boundary is not a symmetric spatial separator.  It is a pair of directed channels whose likelihoods split once the full boundary path is fixed.  Sensory variables carry exterior information into the boundary; active variables carry interior influence out through the boundary.  Conditional on both histories, there is no remaining exterior-interior likelihood term.
\end{remark}

\section{Fixed-time conditional independence}

\subsection{Instantaneous boundaries}

The following definition introduces instantaneous boundaries.

\begin{defn}
For a process $(Y,B,X)$, the boundary variable $B_t$ is an instantaneous Markov boundary at time $t$ if
\[
        Y_t\ci X_t\mid B_t.
\]
\end{defn}

\begin{prop}
Assume $(Y_t,B_t,X_t)$ has a strictly positive smooth density $p_t(y,b,x)$ with respect to a product reference measure.  Then
\[
        Y_t\ci X_t\mid B_t
\]
if and only if, for almost every $b$,
\[
        p_t(y,x\mid b)=p_t(y\mid b)p_t(x\mid b).
\]
If the density is positive and smooth in $(y,x)$, this is equivalent to
\begin{equation}\label{mixed-hessian-zero-eq}
        \partial_{y_a}\partial_{x_r}\log p_t(y,x\mid b)=0
\end{equation}
for all exterior indices $a$ and interior indices $r$, wherever the derivatives are defined.
\end{prop}

\begin{proof}
The first equivalence is the kernel factorisation criterion applied at the single time $t$.  If the conditional density factorises, then
\[
        \log p_t(y,x\mid b)=\log p_t(y\mid b)+\log p_t(x\mid b),
\]
so all mixed exterior-interior derivatives vanish.  Conversely, if the mixed derivatives vanish on a connected coordinate chart, then $\partial_{y_a}\log p_t$ is independent of $x$ for every $a$.  Integrating in the $y$-variables shows that
\[
        \log p_t(y,x\mid b)=u_b(y)+v_b(x)
\]
up to a $b$-dependent constant.  Exponentiating gives factorisation of the conditional density.
\end{proof}

\subsection{Path-wise independence contrasted with time-wise independence}

The following theorem establishes the mixture obstruction from path to time.

\begin{theorem}
Assume the path space boundary condition
\[
        Y_{[0,T]}\ci X_{[0,T]}\mid B_{[0,T]}.
\]
Fix $t\in[0,T]$.  Let $\mathcal B_T=\sigma(B_{[0,T]})$ and $\mathcal B_t=\sigma(B_t)$.  For bounded measurable $f,g$, define
\[
        m_f(\beta)=\bE[f(Y_t)\mid B_{[0,T]}=\beta],
        n_g(\beta)=\bE[g(X_t)\mid B_{[0,T]}=\beta].
\]
Then
\begin{equation}\label{time-mixture-identity-eq}
\begin{aligned}
&\bE[f(Y_t)g(X_t)\mid B_t]\\
&  =\bE[m_f(B_{[0,T]})n_g(B_{[0,T]})\mid B_t].
\end{aligned}
\end{equation}
Consequently $Y_t\ci X_t\mid B_t$ holds if and only if
\begin{equation}\label{mixture-cov-zero-eq}
\Cov\bigl(m_f(B_{[0,T]}),n_g(B_{[0,T]})\mid B_t\bigr)=0
\end{equation}
for all bounded measurable $f,g$.
\end{theorem}

\begin{proof}
The path space boundary condition gives
\[
        \bE[f(Y_t)g(X_t)\mid B_{[0,T]}]
        =m_f(B_{[0,T]})n_g(B_{[0,T]}).
\]
Taking conditional expectation with respect to the smaller sigma-field $\sigma(B_t)$ gives \eqref{time-mixture-identity-eq}.  On the other hand,
\[
        \bE[f(Y_t)\mid B_t]
        =\bE[m_f(B_{[0,T]})\mid B_t]
\]
and similarly for $g(X_t)$.  Therefore the instantaneous conditional independence identity
\[
        \bE[f(Y_t)g(X_t)\mid B_t]
        =\bE[f(Y_t)\mid B_t]\bE[g(X_t)\mid B_t]
\]
is equivalent to vanishing of the conditional covariance \eqref{mixture-cov-zero-eq} for all $f,g$.
\end{proof}

The following corollary gives a sufficient collapse condition.

\begin{corollary}
Under the assumptions of the theorem, if for every bounded $f$ and $g$ the functions
\[
        \bE[f(Y_t)\mid B_{[0,T]}],  \bE[g(X_t)\mid B_{[0,T]}]
\]
are $\sigma(B_t)$-measurable, then
\[
        Y_t\ci X_t\mid B_t.
\]
\end{corollary}

\begin{proof}
If both conditional means are already $\sigma(B_t)$-measurable, then their conditional covariance given $B_t$ is zero.  Apply the mixture obstruction theorem.
\end{proof}

\begin{remark}
This theorem explains the loss caused by passing from histories to time slices.  Conditioning on the whole boundary path can remove dependence, but averaging over possible boundary histories with the same instantaneous value can reintroduce dependence through a mixture.
\end{remark}

\section{Random dynamical systems and cocycles}

\subsection{Cocycles and path laws}

The following definition introduces random dynamical systems.

\begin{defn}
A random dynamical system over a metric dynamical system $(\Omega,\Fil,\prob,\theta_t)$ is a measurable map
\[
        \varphi:\R_+\times\Omega\times M\to M
\]
such that
\[
        \varphi(0,\omega,x)=x,
        \varphi(t+s,\omega,x)=\varphi(t,\theta_s\omega,\varphi(s,\omega,x)).
\]
\end{defn}

The following proposition shows that conditioning on noise gives a Dirac trajectory.

\begin{prop}
Fix $x\in M$.  Define the solution-path map
\[
        \Phi_x(\omega)=\bigl(t\mapsto \varphi(t,\omega,x)\bigr)
        \in\Path_T(M).
\]
The path law of the solution is $(\Phi_x)_*\prob$.  Conditional on the noise path $\omega$, the state trajectory is deterministic:
\[
        \Law(X_{[0,T]}\mid \omega)=\delta_{\Phi_x(\omega)}.
\]
\end{prop}

\begin{proof}
The first statement is the definition of the pushforward law.  Since $X_{[0,T]}=\Phi_x(\omega)$ as a measurable function of $\omega$, the regular conditional law of $X_{[0,T]}$ given $\omega$ is the Dirac mass at that value.
\end{proof}

\subsection{Invariant skew products}

The following definition introduces invariant measures for an random dynamical system.

\begin{defn}
A probability measure $\rho$ on $\Omega\times M$ is invariant for the skew product
\[
        \Theta_t(\omega,x)=(\theta_t\omega,\varphi(t,\omega,x))
\]
if $(\Theta_t)_*\rho=\rho$ for all $t\geqslant0$.  If its $\Omega$-marginal is $\prob$, a disintegration has the form
\[
        \rho(\dd\omega,\dd x)=\prob(\dd\omega)\rho_\omega(\dd x).
\]
\end{defn}

The following proposition describes transport of sample measures.

\begin{prop}
Let $\rho(\dd\omega,\dd x)=\prob(\dd\omega)\rho_\omega(\dd x)$ be invariant for the skew product.  Then
\begin{equation}\label{sample-measure-transport-eq}
        \varphi(t,\omega,\cdot)_*\rho_\omega=\rho_{\theta_t\omega}
\end{equation}
for $\prob$-almost every $\omega$ and every $t\geqslant0$, after modification on a null set.
\end{prop}

\begin{proof}
For bounded measurable $F(\omega,x)$, invariance gives
\[
\begin{aligned}
        \int F(\omega,x)\,\prob(\dd\omega)\rho_\omega(\dd x)
        &=\int F(\theta_t\omega,\varphi(t,\omega,x))\,
          \prob(\dd\omega)\rho_\omega(\dd x).
\end{aligned}
\]
Using $\theta_t$-invariance of $\prob$ and disintegrating over the first coordinate, the conditional measure over $\theta_t\omega$ must be the pushforward of $\rho_\omega$ by $\varphi(t,\omega,\cdot)$.  Essential uniqueness of disintegration gives \eqref{sample-measure-transport-eq}.
\end{proof}

The following definition introduces random path space boundaries.

\begin{defn}
Assume $M=M_Y\times M_B\times M_X$.  A random invariant measure $\rho_\omega$ has a fibrewise Markov boundary if, for $\prob$-almost every $\omega$,
\[
        \rho_\omega(\dd y,\dd b,\dd x)
        =\rho_{B,\omega}(\dd b)\rho_{Y\mid B,\omega}^b(\dd y)
          \rho_{X\mid B,\omega}^b(\dd x).
\]
It has a random path space boundary over $[0,T]$ if the corresponding stationary cocycle path law satisfies the analogous factorisation conditional on $B_{[0,T]}$ and, when retained, on the base noise $\omega$.
\end{defn}

The following theorem establishes random dynamical system boundary factorisation.

\begin{theorem}
Let $\varphi$ be a cocycle on $M_Y\times M_B\times M_X$ with invariant skew-product measure $\rho(\dd\omega,\dd x)=\prob(\dd\omega)\rho_\omega(\dd x)$.  Start the system with $Z_0\sim\rho_\omega$ over the base $\omega$.  If for $\prob$-almost every $\omega$ the sample measure factorises conditionally over $B$, and if the cocycle respects the boundary in the sense that, for fixed $(\omega,b_{[0,T]})$, the exterior and interior solution maps depend on independent fibre variables and no common residual variable outside the boundary, then
\[
        Y_{[0,T]}\ci X_{[0,T]}\mid B_{[0,T]}
\]
under the stationary path law.  Conversely, if the stationary path map is fibrewise injective up to null sets in the initial exterior and interior coordinates, then stationary path space factorisation implies fibrewise conditional factorisation of $\rho_\omega$.
\end{theorem}

\begin{proof}
For fixed $\omega$, the path map sends an initial condition $(y_0,b_0,x_0)$ to a path $(y,b,x)$.  By assumption, after conditioning on the boundary path $b$, the exterior path is a measurable function of the exterior fibre variable and the fixed data $(\omega,b)$, while the interior path is a measurable function of the interior fibre variable and the fixed data $(\omega,b)$.  The conditional sample measure of the initial exterior and interior variables factorises over the boundary variable.  Measurable pushforwards preserve product structure, so the conditional path law factorises over $(\omega,b)$.  Integrating over the conditional law of $\omega$ given $b$ preserves the factorisation exactly when no common residual dependence through $\omega$ remains outside the observed boundary path; this is the stated cocycle-respecting condition.

For the converse, suppose the stationary path law factorises and the path map is fibrewise injective in the initial exterior and interior coordinates.  Pull the conditional path factorisation back through the inverse path map.  Since injective measurable maps between standard Borel spaces admit measurable inverses on their images, the product disintegration of paths induces a product disintegration of the initial fibre measures.  Hence $\rho_\omega$ factorises conditionally over $B$.
\end{proof}

\begin{remark}
The random dynamical system formulation separates two kinds of conditioning.  Conditioning on the noise $\omega$ makes trajectories deterministic.  Conditioning on the boundary path asks whether, after the observable boundary history has been fixed, the remaining exterior and interior uncertainties are product uncertainties.
\end{remark}

\section{Stationary factorisation and convergence}

\subsection{Equilibrium}

Consider a diffusion on $\R^{d_Y+d_B+d_X}$ with invariant density
\[
        \rho_*(y,b,x)=Z^{-1}e^{-\Phi(y,b,x)}.
\]
In equilibrium gradient Langevin form, the current vanishes and the stationary law determines the reversible drift through the potential.

The following theorem establishes the equilibrium potential-splitting criterion.

\begin{theorem}
Assume $\rho_*(y,b,x)>0$.  Then
\[
        Y\ci X\mid B
\]
under the invariant law if and only if there exist functions $\Phi_Y,\Phi_X,\Phi_B$ such that
\begin{equation}\label{potential-split-eq}
        \Phi(y,b,x)=\Phi_Y(y,b)+\Phi_X(x,b)+\Phi_B(b)
\end{equation}
up to an additive constant and null sets.  If $\Phi$ is $C^2$ in $(y,x)$ on connected fibres, this is equivalent to
\begin{equation}\label{equilibrium-cross-hessian-eq}
        \partial_{y_a}\partial_{x_r}\Phi(y,b,x)=0
\end{equation}
for all $a,r$.
\end{theorem}

\begin{proof}
The invariant conditional density is
\[
        \rho_*(y,x\mid b)=\frac{e^{-\Phi(y,b,x)}}{
        \int e^{-\Phi(y',b,x')}\,\dd y'\dd x'}.
\]
Conditional independence is equivalent to factorisation of this density in $(y,x)$.  Since the exponential is positive, this is equivalent to additive separation of $-\Phi$, hence of $\Phi$, up to a function of $b$ absorbed into $\Phi_B$.  If $\Phi$ is $C^2$, additive separation implies the mixed Hessian condition.  Conversely, vanishing mixed Hessian implies $\partial_{y_a}\Phi$ is independent of $x$, so integrating over $y$ and then over $x$ gives \eqref{potential-split-eq} on connected fibres.
\end{proof}

\subsection{Nonequilibrium stationary currents}

For a diffusion
\[
        \dd X_t=b(X_t)\,\dd t+\sigma(X_t)\,\dd W_t,
          D=\frac12\sigma\sigma^\top,
\]
with invariant density $\rho$, the stationary Fokker--Planck equation is
\[
        0=-\nabla\cdot(b\rho)+\nabla\cdot\nabla\cdot(D\rho).
\]
The stationary current is
\[
        J=b\rho-\nabla\cdot(D\rho),
\]
and stationarity is $\nabla\cdot J=0$.  Equilibrium is the special case $J=0$.  A nonequilibrium steady state has $J\neq0$ but $\nabla\cdot J=0$.

The following theorem establishes the NESS boundary criterion.

\begin{theorem}
Let a non-degenerate diffusion have smooth positive invariant density $\rho_*$ and stationary current $J_*$.  Its invariant law has an instantaneous Markov boundary,
\[
        Y\ci X\mid B,
\]
if and only if
\begin{equation}\label{ness-density-factor-eq}
        \rho_*(y,b,x)=\rho_B(b)\rho_{Y\mid B}(y\mid b)\rho_{X\mid B}(x\mid b).
\end{equation}
The process has a stationary path space boundary over $[0,T]$ if, in addition, the stationary path law satisfies
\[
        \Law(Y_{[0,T]},X_{[0,T]}\mid B_{[0,T]}=b)
        =\Law(Y_{[0,T]}\mid B_{[0,T]}=b)\otimes
         \Law(X_{[0,T]}\mid B_{[0,T]}=b)
\]
for $P_B$-almost every $b$.  A sufficient generator-level condition for this stronger statement is the conditional generator splitting \eqref{generator-splitting-eq} under the stationary clamped law for almost every boundary path.
\end{theorem}

\begin{proof}
The first statement is exactly the fixed-time kernel factorisation criterion applied to the invariant density.  For the path space statement, apply the path space factorisation theorem to the stationary path measure.  The generator-level condition is sufficient by the generator splitting theorem: if, after conditioning on the boundary path, the exterior and interior martingale problem splits and is unique, the conditional path law is a product.
\end{proof}

The following proposition shows convergence to a factorised steady state.

\begin{prop}
Let $(X_t)_{t\geqslant0}$ be a Markov process with invariant probability measure $\mu$.  Suppose
\[
        \Law(X_t)\Longrightarrow \mu
\]
as $t\to\infty$.  If $\mu$ satisfies $Y\ci X\mid B$, then every weak limit of the time-slice laws satisfies the instantaneous boundary factorisation.  If the process is stationary and the stationary path law satisfies the path space boundary condition, then every finite window sampled from stationarity satisfies the same path space factorisation.
\end{prop}

\begin{proof}
The first assertion is immediate from weak convergence to $\mu$ together with the assumption that the limiting invariant law factorises conditionally.  More explicitly, along any convergent subsequence the limit is $\mu$, and its disintegration over $B$ has product conditional kernels.  The second assertion is just stationarity: every finite window has the same law as the window starting at time zero under the stationary path measure, so the same disintegration identity holds.
\end{proof}

\subsection{Ao or Helmholtz form}

The following proposition gives the finite-dimensional stationary Helmholtz decomposition.

\begin{prop}
Let $\rho>0$ be a smooth invariant density for a diffusion with diffusion tensor $D$.  Define
\[
        J=b\rho-\nabla\cdot(D\rho).
\]
Then $\nabla\cdot J=0$ and
\begin{equation}\label{helmholtz-finite-eq}
        b=\rho^{-1}\nabla\cdot(D\rho)+\rho^{-1}J.
\end{equation}
The first term is reversible with respect to $
ho$, while the second term is stationary-current preserving.  If $\rho=e^{-\Phi}/Z$, the reversible part is
\[
        b_{\mathrm{rev}}=\nabla\cdot D-D\nabla\Phi.
\]
\end{prop}

\begin{proof}
The stationary Fokker--Planck equation is exactly $\nabla\cdot J=0$.  Rearranging the definition of $J$ gives \eqref{helmholtz-finite-eq}.  The generator associated with $\rho^{-1}\nabla\cdot(D\rho)$ is symmetric in $L^2(\rho\dd x)$ by integration by parts.  The vector field $\rho^{-1}J$ is skew at the density level because $\nabla\cdot J=0$.  Substituting $\rho=e^{-\Phi}/Z$ gives $\rho^{-1}\nabla\cdot(D\rho)=\nabla\cdot D-D\nabla\Phi$.
\end{proof}

The following proposition gives the boundary-compatible Ao sparsity criterion.

\begin{prop}
Suppose a stationary diffusion admits an Ao-type representation
\[
        b=-(\Gamma-Q)\nabla\Phi+\nabla\cdot(\Gamma-Q),
          \Gamma=\Gamma^\top\geqslant0,
          Q=-Q^\top,
\]
with invariant density $\rho\propto e^{-\Phi}$.  A sufficient condition for the invariant law to satisfy $Y\ci X\mid B$ and for the generator to contain no direct exterior-interior second-order coupling is:
\begin{enumerate}
\item $\Phi(y,b,x)=\Phi_Y(y,b)+\Phi_X(x,b)+\Phi_B(b)$;
\item the $(Y,X)$-blocks of $\Gamma$ vanish;
\item the $(Y,X)$-blocks of the solenoidal coupling $Q$ vanish or are mediated only through $B$ in the sense that the conditional generator given $B$ splits.
\end{enumerate}
\end{prop}

\begin{proof}
The first condition gives invariant-density factorisation by the equilibrium potential-splitting criterion, independently of whether the stationary current vanishes.  The second condition removes direct mixed second-order diffusion between $Y$ and $X$.  The third removes direct solenoidal transport coupling between exterior and interior variables after the boundary is fixed.  Under these sparsity conditions, the conditional generator has no mixed exterior-interior second-order term and no drift coefficient depending on the opposite variable except through the fixed boundary.  The generator splitting criterion then gives conditional path factorisation when the corresponding martingale problems are unique.
\end{proof}

\section{Conditional mutual information and likelihood ratios}

The following definition introduces the conditional information gap on path space.

\begin{defn}
For a joint path law $P_T$, define
\[
        \mathcal G_T
        :=\CMI(Y_{[0,T]};X_{[0,T]}\mid B_{[0,T]}).
\]
Equivalently,
\[
        \mathcal G_T
        =\int \Ent\bigl(P^b_{Y,X}\mid P^b_Y\otimes P^b_X\bigr)P_B(\dd b).
\]
\end{defn}

The following theorem establishes the product projection property of conditional mutual information.

\begin{theorem}
The conditional information gap satisfies
\[
        \mathcal G_T\geqslant0,
\]
with equality if and only if $B_{[0,T]}$ is a path space Markov boundary between $Y_{[0,T]}$ and $X_{[0,T]}$.  Moreover, for each fixed $b$, the product measure closest to the true conditional law $P^b_{Y,X}$ in the inclusive relative entropy
\[
        \Ent(P^b_{Y,X}\mid Q_Y\otimes Q_X)
\]
over product measures $Q_Y\otimes Q_X$ is the product of the true conditional marginals $P^b_Y\otimes P^b_X$, and the minimum value is
\[
        \Ent(P^b_{Y,X}\mid P^b_Y\otimes P^b_X).
\]
\end{theorem}

\begin{proof}
Non-negativity and the equality condition follow from the conditional relative entropy identity.  For the projection statement, fix $b$ and let $P=P^b_{Y,X}$.  For any product $Q_Y\otimes Q_X$, the chain rule gives
\[
\begin{aligned}
        \Ent(P\mid Q_Y\otimes Q_X)
        &=\int \log\frac{\dd P}{\dd(P_Y\otimes P_X)}\,\dd P
          +\int \log\frac{\dd(P_Y\otimes P_X)}{\dd(Q_Y\otimes Q_X)}\,\dd P\\
        &=\Ent(P\mid P_Y\otimes P_X)+\Ent(P_Y\mid Q_Y)+\Ent(P_X\mid Q_X).
\end{aligned}
\]
The last two terms are non-negative and vanish exactly at $Q_Y=P_Y$, $Q_X=P_X$.  Hence the closest product approximation in this direction is the product of marginals, and the minimal gap is the mutual information.
\end{proof}

The following theorem records the Gibbs identity associated with a path likelihood.

\begin{theorem}
Fix a boundary path $b$ and write the exact conditional density as
\[
        p(y,x\mid b)=\frac{p(y,x,b)}{p(b)}.
\]
For any density $q(y,x)$, define
\[
        \mathcal F_b(q)
        :=\bE_q[-\log p(Y,X,b)]-\mathrm H(q),
\]
where $\mathrm H(q)=-\bE_q\log q$ is Shannon entropy.  Then
\begin{equation}\label{vfe-identity-eq}
        \mathcal F_b(q)+\log p(b)=\Ent(q\mid p(\cdot,\cdot\mid b)).
\end{equation}
Consequently minimising $\mathcal F_b$ over all $q$ recovers the exact conditional law, while minimising over product densities $q_Yq_X$ gives the best factorised conditional law.  If the product family contains the conditional marginals, the excess of the best product free energy over the exact free energy is the conditional mutual information at $b$.
\end{theorem}

\begin{proof}
Compute
\[
\begin{aligned}
        \Ent(q\mid p(\cdot,\cdot\mid b))
        &=\int q(y,x)\log\frac{q(y,x)}{p(y,x\mid b)}\,\dd y\dd x\\
        &=\int q\log q\,\dd y\dd x-\int q\log p(y,x,b)\,\dd y\dd x+\int q\log p(b)\,\dd y\dd x\\
        &=\mathcal F_b(q)+\log p(b).
\end{aligned}
\]
Since relative entropy is non-negative and vanishes exactly at equality of measures, the unconstrained minimiser is the exact conditional law.  Restricting to product densities gives the product projection.  Applying the previous theorem to the exact conditional law shows that the product of conditional marginals is the inclusive relative entropy product projection and that its gap is the conditional mutual information.
\end{proof}

The following corollary gives the likelihood-ratio interpretation.

\begin{corollary}
Let
\[
        R_b(y,x)=\frac{\dd P^b_{Y,X}}{\dd(P^b_Y\otimes P^b_X)}(y,x)
\]
where this derivative exists.  Then
\[
        \mathcal G_T=\bE[\log R_{B_{[0,T]}}(Y_{[0,T]},X_{[0,T]})].
\]
Thus making the likelihood ratio close to one in relative entropy is exactly minimising the path space factorisation gap.  In particular, Pinsker's inequality gives
\[
        \bE_{B}\|P^b_{Y,X}-P^b_Y\otimes P^b_X\|_{\mathrm{TV}}^2
        \leqslant \frac12\mathcal G_T.
\]
\end{corollary}

\begin{proof}
The identity is the definition of conditional relative entropy.  Pinsker's inequality gives
\[
        \|P^b_{Y,X}-P^b_Y\otimes P^b_X\|_{\mathrm{TV}}^2
        \leqslant \frac12\Ent(P^b_{Y,X}\mid P^b_Y\otimes P^b_X)
\]
for each $b$.  Integrating over $P_B$ proves the bound.
\end{proof}

\section{Mixing, distinguishability, and individuation}

The following definition introduces distinguishability of coupled and separated path ensembles.

\begin{defn}
The distinguishability between the true conditional path ensemble and the separated ensemble is
\[
        \mathcal D_T
        :=\Ent\bigl(P_{Y,B,X}\mid P_{Y\mid B}\otimes P_{X\mid B}\otimes P_B\bigr)
        =\mathcal G_T.
\]
\end{defn}

The following proposition gives data processing for observed paths.

\begin{prop}
Let $O=\Psi(Y_{[0,T]},X_{[0,T]},B_{[0,T]})$ be any measurable observation.  Then the distinguishability between the observed true law and the observed separated law is at most $\mathcal G_T$.  In particular, any coarse-grained observer can distinguish the coupled model from the separated model with relative-entropy power no larger than the path space factorisation gap.
\end{prop}

\begin{proof}
Let $P$ be the true joint path law and let
\[
        Q=P_{Y\mid B}\otimes P_{X\mid B}\otimes P_B
\]
be the separated law.  Then $\mathcal G_T=\Ent(P\mid Q)$.  The observation map $\Psi$ pushes these measures forward to $\Psi_*P$ and $\Psi_*Q$.  Relative entropy decreases under measurable pushforward, so
\[
        \Ent(\Psi_*P\mid \Psi_*Q)\leqslant \Ent(P\mid Q)=\mathcal G_T.
\]
\end{proof}

The following theorem compares mixing with boundary distinguishability.

\begin{theorem}
Assume the stationary path law of $(Y,B,X)$ is ergodic.  Let $P_T$ be the length-$T$ stationary path law and $Q_T=P_{Y\mid B,T}\otimes P_{X\mid B,T}\otimes P_{B,T}$ the separated path law.  Then
\[
        \frac1T\Ent(P_T\mid Q_T)
\]
is the asymptotic per-time distinguishability rate between the true coupled stationary process and the process in which exterior and interior histories are conditionally separated over the boundary.  If this rate vanishes along a sequence $T_n\to\infty$, no test based on those windows can distinguish the coupled and separated ensembles with exponential error rate bounded away from zero.
\end{theorem}

\begin{proof}
The first statement is a definition of relative-entropy rate for stationary processes.  For the testing statement, the Chernoff--Stein lemma says that, in simple hypothesis testing between $P_T$ and $Q_T$, the optimal type-II error exponent at fixed type-I error is governed by $T^{-1}\Ent(P_T\mid Q_T)$, provided the stationary ergodic relative-entropy rate exists.  If the rate tends to zero along a subsequence, the exponential distinguishability along those windows vanishes.  Even without invoking the full Chernoff--Stein theorem, Pinsker's inequality implies that small relative entropy gives small total variation on average, hence no uniformly powerful bounded test can separate the ensembles on those windows.
\end{proof}

\begin{remark}
This is the formal content of boundary individuation.  The exterior and interior may still be strongly coupled through the boundary.  What becomes small is the residual distinguishability between the true path ensemble and the ensemble in which all remaining exterior-interior dependence has been removed after the boundary history is fixed.
\end{remark}

\subsection{Time-reversal distinguishability and housekeeping heat}

The preceding relative entropy compares the true coupled path ensemble with the conditionally separated ensemble.  Stochastic thermodynamics introduces a second, logically independent, comparison: the comparison of a stationary path ensemble with its time reversal.  These two comparisons should not be conflated.  The boundary gap measures residual exterior-interior mixing after the boundary has been observed.  Housekeeping heat measures the cost of maintaining a time-oriented nonequilibrium current, even when the conditional boundary factorisation is exact.

Let $\Theta_T$ denote the time-reversal map on path space,
\[
        (\Theta_T x)_s=x_{T-s},  0\leqslant s\leqslant T,
\]
possibly composed with a parity involution when momenta or odd variables are present.  Given a stationary law $P_T$ on paths, write
\[
        P_T^\dagger := (\Theta_T)_*P_T.
\]

The following definition introduces two distinguishability functionals.

\begin{defn}
Let
\[
        Q_T=P_{Y\mid B,T}\otimes P_{X\mid B,T}\otimes P_{B,T}
\]
be the conditionally separated path law.  The \emph{boundary distinguishability} is
\[
        \mathcal G_T=\Ent(P_T\mid Q_T).
\]
The \emph{time-reversal distinguishability} is
\[
        \mathcal E_T=\Ent(P_T\mid P_T^\dagger).
\]
When the limit exists, the entropy production rate is
\[
        \ep=\lim_{T\to\infty}\frac1T\mathcal E_T.
\]
In thermodynamic units the corresponding entropy production rate is $k_B\ep$.  If the environment is an isothermal bath at temperature $T_{\mathrm{bath}}$, the housekeeping heat rate at a stationary nonequilibrium state is
\[
        \dot Q_{\mathrm{hk}}=k_B T_{\mathrm{bath}}\,\ep,
\]
with the convention that $\ep$ is measured in nats per unit time.
\end{defn}

The following proposition distinguishes two orthogonal notions of indistinguishability.

\begin{prop}
The identities
\[
        \mathcal G_T=\Ent(P_T\mid Q_T), 
        \mathcal E_T=\Ent(P_T\mid P_T^\dagger)
\]
compare $P_T$ with two different reference laws.  The first vanishes exactly when
\[
        Y_{[0,T]}\ci X_{[0,T]}\mid B_{[0,T]},
\]
whereas the second vanishes exactly when the stationary path law is invariant under time reversal.  Consequently an exact path space boundary does not imply zero entropy production, and zero entropy production does not imply a boundary factorisation.
\end{prop}

\begin{proof}
The first statement is the defining property of conditional mutual information: $\Ent(P_T\mid Q_T)=0$ if and only if $P_T=Q_T$, which is precisely conditional factorisation over $B_{[0,T]}$.  The second statement follows from non-negativity of relative entropy: $\Ent(P_T\mid P_T^\dagger)=0$ if and only if $P_T=P_T^\dagger$.  These equalities involve different comparison measures.  There is no implication between them in general.  For example, a product of two independent driven diffusions can have an exact conditional product structure and positive stationary current.  Conversely, an equilibrium Gaussian vector can be time reversible while failing to factorise conditionally over a proposed boundary.
\end{proof}

The following remark explains the individuation interpretation.

\begin{remark}
The pair $(T^{-1}\mathcal G_T,T^{-1}\mathcal E_T)$ separates two senses in which a system can be individuated.  Small $T^{-1}\mathcal G_T$ says that, after the boundary history is observed, exterior and interior histories are hard to distinguish from conditionally separated histories.  Positive $T^{-1}\mathcal E_T$ says that the individuated system is maintained by a time-oriented current.  Thus an equilibrium boundary has both quantities small.  A nonequilibrium steady-state boundary has small boundary distinguishability but positive housekeeping heat.  A poorly individuated mixing regime has large boundary distinguishability, whether or not it is time reversible.
\end{remark}

\subsection{Stochastic entropy production and fluctuation identity}

The following definition introduces stochastic entropy production.

\begin{defn}
Assume $P_T\ll P_T^\dagger$.  The total stochastic entropy production along a realised path is
\[
        \Sigma_T(x):=\log\frac{\dd P_T}{\dd P_T^\dagger}(x).
\]
If $\nu_t$ is the one-time law and $\mu$ is a reference stationary law, the relative system entropy is $-\log(\dd\nu_t/\dd\mu)(X_t)$, and the medium entropy production is obtained by subtracting the boundary change of this system entropy from $\Sigma_T$.  At stationarity the expected boundary term vanishes, so the mean total entropy production rate equals the mean medium entropy production rate.
\end{defn}

The following theorem establishes the integral fluctuation theorem and the second law.

\begin{theorem}
If $P_T\ll P_T^\dagger$, then
\[
        \bE_{P_T}\exp(-\Sigma_T)=1,
\]
and hence
\[
        \bE_{P_T}\Sigma_T=\Ent(P_T\mid P_T^\dagger)\geqslant0.
\]
If $(P_T)_{T>0}$ is stationary and the entropy production rate exists, then
\[
        \lim_{T\to\infty}\frac1T\bE\Sigma_T=\ep.
\]
\end{theorem}

\begin{proof}
Since $\Sigma_T=\log(\dd P_T/\dd P_T^\dagger)$,
\[
        \bE_{P_T}e^{-\Sigma_T}
        =\int \frac{\dd P_T^\dagger}{\dd P_T}\,\dd P_T
        =\int \dd P_T^\dagger=1.
\]
The identity $\bE\Sigma_T=\Ent(P_T\mid P_T^\dagger)$ is the definition of relative entropy.  Non-negativity follows from Jensen's inequality or from Gibbs' inequality.  Dividing by $T$ and passing to the stationary limit gives the last formula.
\end{proof}

The following proposition gives data processing for thermodynamic observations.

\begin{prop}
Let $O=\Psi(X_{[0,T]})$ be any measurable observation of the full path.  Then
\[
        \Ent(\Psi_*P_T\mid \Psi_*P_T^\dagger)\leqslant \Ent(P_T\mid P_T^\dagger).
\]
Thus entropy production is the maximal relative-entropy distinguishability, among all observations of the path, between the forward movie and the reversed movie.
\end{prop}

\begin{proof}
Apply monotonicity of the relative entropy under measurable pushforward to $P_T$ and $P_T^\dagger$.
\end{proof}

\subsection{Housekeeping heat for stationary diffusions}

We now give the finite-dimensional formula which should be read as the thermodynamic companion to the boundary distinguishability formula.  Let
\begin{equation}\label{st-diffusion-eq}
        \dd X_t=b(X_t)\,\dd t+\sigma(X_t)\,\dd W_t,
          D=\frac12\sigma\sigma^\top,
\end{equation}
be a diffusion on $\R^d$.  Suppose it has a smooth positive invariant density $\pi$, and assume either no boundary or boundary conditions removing all integration-by-parts terms.  The stationary probability current is
\begin{equation}\label{stationary-current-eq}
        J_\pi=b\pi-\nabla\cdot(D\pi),
          \nabla\cdot J_\pi=0.
\end{equation}
The current velocity is $v_\pi=J_\pi/\pi$.

The following theorem establishes stationary entropy production and housekeeping heat.

\begin{theorem}
Assume $v_\pi(x)\in\Range D(x)$ for $\pi$-almost every $x$, and
\[
        \int v_\pi^\top D^\dagger v_\pi\,\pi\,\dd x<\infty.
\]
Let $P_T^\pi$ be the stationary path law of \eqref{st-diffusion-eq}.  Then
\begin{equation}\label{epr-current-eq}
        \ep
        =\lim_{T\to\infty}\frac1T
          \Ent(P_T^\pi\mid (P_T^\pi)^\dagger)
        =\int_{\R^d}\frac{J_\pi^\top D^\dagger J_\pi}{\pi}\,\dd x.
\end{equation}
Consequently, in an isothermal environment,
\[
        \dot Q_{\mathrm{hk}}
        =k_B T_{\mathrm{bath}}
          \int_{\R^d}\frac{J_\pi^\top D^\dagger J_\pi}{\pi}\,\dd x.
\]
If the range condition fails on a set of positive $\pi$-measure, the forward and reversed stationary path laws are singular at the corresponding local level and the entropy production is infinite.
\end{theorem}

\begin{proof}
The time-reversed stationary diffusion has the same diffusion tensor and drift
\[
        b^\dagger=b-2v_\pi.
\]
Indeed the current associated with $b^\dagger$ is $-J_\pi$, while the invariant density remains $\pi$.  The drift difference between the forward and reversed equations is therefore $2v_\pi$.  If $v_\pi\in\Range D$, this drift difference is a Cameron--Martin/Girsanov shift of the driving noise.  The relative entropy per unit time for changing a diffusion drift by $2v_\pi$ is
\[
        \frac12\int \bigl|\sigma^\dagger(2v_\pi)\bigr|^2\pi\,\dd x.
\]
Since $\sigma\sigma^\top=2D$, the Moore--Penrose identity gives
\[
        \frac12\bigl|\sigma^\dagger(2v_\pi)\bigr|^2
        =v_\pi^\top D^\dagger v_\pi.
\]
Stationarity then yields \eqref{epr-current-eq}.  If the range condition fails, the forward and reversed laws differ by a drift component outside the Cameron--Martin directions generated by the noise.  The Cameron--Martin support dichotomy gives singularity rather than a finite likelihood ratio, so the relative entropy is infinite.
\end{proof}

The following remark explains what the formula says.

\begin{remark}
The current $J_\pi$ preserves the stationary density because $\nabla\cdot J_\pi=0$.  It is therefore conservative at the density level.  Nevertheless it produces entropy because it gives the stationary process an orientation in time.  Housekeeping heat is exactly the heat required to maintain this orientation.  Cancelling $J_\pi$ removes the stationary time arrow and removes the housekeeping heat, even though it need not change the invariant density.
\end{remark}

\subsection{Excess dissipation and housekeeping dissipation}

The stationary formula is the special case of a decomposition valid away from stationarity.  Let $p_t$ be the time-dependent density, write $q_t=p_t/\pi$, and define the instantaneous current
\[
        J_t=bp_t-\nabla\cdot(Dp_t).
\]
Using \eqref{stationary-current-eq}, one has the algebraic splitting
\begin{equation}\label{current-splitting-hk-eq}
        J_t=q_tJ_\pi-p_tD\nabla\log q_t.
\end{equation}
The first term is the transported stationary current.  The second term is the relaxation current down the relative-entropy gradient.

The following theorem establishes the housekeeping-excess decomposition.

\begin{theorem}
Assume $D$ is symmetric positive definite, the preceding integration by parts is justified, and all terms below are finite.  Define
\[
        \sigma_t=\int \frac{J_t^\top D^{-1}J_t}{p_t}\,\dd x,
\]
\[
        \sigma_t^{\mathrm{hk}}
        =\int p_t\, v_\pi^\top D^{-1}v_\pi\,\dd x
        =\int q_t\frac{J_\pi^\top D^{-1}J_\pi}{\pi}\,\dd x,
\]
and
\[
        \sigma_t^{\mathrm{ex}}
        =\int p_t\,\nabla\log q_t^\top D\nabla\log q_t\,\dd x.
\]
Then
\begin{equation}\label{hk-excess-decomp-eq}
        \sigma_t=\sigma_t^{\mathrm{hk}}+\sigma_t^{\mathrm{ex}}.
\end{equation}
Moreover
\begin{equation}\label{freeenergy-dissipation-eq}
        \frac{\dd}{\dd t}\Ent(p_t\mid\pi)
        =-\sigma_t^{\mathrm{ex}}.
\end{equation}
At stationarity, $p_t=\pi$, the excess term vanishes and
\[
        \sigma_t^{\mathrm{hk}}=\ep.
\]
\end{theorem}

\begin{proof}
Substitute \eqref{current-splitting-hk-eq} into $\sigma_t$.  Since $p_t=q_t\pi$,
\[
\begin{aligned}
        \sigma_t
        &=\int \frac{(q_tJ_\pi-p_tD\nabla\log q_t)^\top D^{-1}
        (q_tJ_\pi-p_tD\nabla\log q_t)}{p_t}\,\dd x\\
        &=\int q_t\frac{J_\pi^\top D^{-1}J_\pi}{\pi}\,\dd x
          +\int p_t\nabla\log q_t^\top D\nabla\log q_t\,\dd x
          -2\int J_\pi\cdot\nabla q_t\,\dd x.
\end{aligned}
\]
The cross term vanishes because $\nabla\cdot J_\pi=0$ and the boundary contribution is zero:
\[
        \int J_\pi\cdot\nabla q_t\,\dd x
        =-\int q_t\nabla\cdot J_\pi\,\dd x=0.
\]
This proves \eqref{hk-excess-decomp-eq}.  For the entropy identity, use the Fokker--Planck equation $\partial_t p_t=-\nabla\cdot J_t$:
\[
\begin{aligned}
        \frac{\dd}{\dd t}\Ent(p_t\mid\pi)
        &=\int (1+\log q_t)\partial_t p_t\,\dd x\\
        &=-\int (1+\log q_t)\nabla\cdot J_t\,\dd x\\
        &=\int \nabla\log q_t\cdot J_t\,\dd x.
\end{aligned}
\]
Using \eqref{current-splitting-hk-eq}, the stationary-current term again vanishes by $\nabla\cdot J_\pi=0$, while the gradient term gives
\[
        -\int p_t\nabla\log q_t^\top D\nabla\log q_t\,\dd x.
\]
This is \eqref{freeenergy-dissipation-eq}.  At $p_t=\pi$, $q_t=1$, so $\sigma_t^{\mathrm{ex}}=0$ and $\sigma_t^{\mathrm{hk}}$ reduces to \eqref{epr-current-eq}.
\end{proof}

The following remark explains why this belongs in the distinguishability section.

\begin{remark}
The excess term is relaxation distinguishability: it is the rate at which the current law becomes less distinguishable from the stationary law $\pi$.  The housekeeping term is time-orientation distinguishability: it remains after relaxation has ended.  Therefore an individuated nonequilibrium system can be perfectly steady, and can even have an exact path space boundary, while still dissipating housekeeping heat through the currents that maintain its separation from the environment.
\end{remark}

\subsection{Counter-rotational environmental forcing}

The preceding formula also shows the precise sense in which an environmental or control force can oppose a nonequilibrium circulation.  Suppose a stationary density $\pi$ is held fixed and an additional drift $u$ changes the stationary current to
\[
        J_\pi^u=J_\pi+\pi u,
\]
with $\nabla\cdot J_\pi^u=0$.  The new housekeeping entropy production is
\[
        \ep(u)=\int \frac{(J_\pi+\pi u)^\top D^\dagger(J_\pi+\pi u)}{\pi}\,\dd x.
\]

The following theorem establishes cancellation of housekeeping heat by a counter-current.

\begin{theorem}
Assume $u$ is an admissible current-preserving perturbation and all terms are finite.  Then
\begin{equation}\label{ep-control-expansion-eq}
        \ep(u)=\ep(0)
        +2\int u^\top D^\dagger J_\pi\,\dd x
        +\int \pi\,u^\top D^\dagger u\,\dd x.
\end{equation}
Consequently $u$ reduces housekeeping heat to first order precisely when
\[
        \int u^\top D^\dagger J_\pi\,\dd x<0.
\]
Over the unrestricted affine class of perturbations preserving $\pi$, the unique minimiser in the noise metric is
\[
        u_*=-\frac{J_\pi}{\pi},
\]
for which $J_\pi^{u_*}=0$ and $\ep(u_*)=0$.
\end{theorem}

\begin{proof}
Expand the square:
\[
\begin{aligned}
        \ep(u)
        &=\int \frac{J_\pi^\top D^\dagger J_\pi}{\pi}\,\dd x
          +2\int u^\top D^\dagger J_\pi\,\dd x
          +\int \pi u^\top D^\dagger u\,\dd x.
\end{aligned}
\]
This is \eqref{ep-control-expansion-eq}.  The first variation at $u=0$ is the middle term.  Completing the square gives
\[
        \ep(u)=\int \pi\left(u+\frac{J_\pi}{\pi}\right)^\top
        D^\dagger
        \left(u+\frac{J_\pi}{\pi}\right)\dd x,
\]
up to the zero constant obtained after cancellation.  Hence the unconstrained minimiser is $u_*=-J_\pi/\pi$.  This perturbation sets the current to zero and therefore removes housekeeping entropy production.
\end{proof}

The following remark discusses ports, boundaries, and admissibility.

\begin{remark}
The theorem is an algebraic statement about a fixed density and an admissible current perturbation.  In an open system the admissible $u$'s are not arbitrary: they may enter only through sensory, active, reservoir, or boundary ports.  Then the optimal counter-current is the orthogonal projection of $-J_\pi/\pi$ onto the allowed control subspace in the Hilbert space $L^2(\pi;D^\dagger)$.  This is the precise version of the statement that an environment can reduce housekeeping heat when it pushes against the stationary circulation.
\end{remark}

\subsection{Factorisation gap, excess heat, and architectural housekeeping}

We now make explicit the relation between the factorisation gap and heat dissipation.  An information theoretic error becomes thermodynamic after one specifies a Gibbs path law and local detailed balance.  Under these hypotheses, the same relative entropy which measures the failure of conditional independence is the excess nonequilibrium free energy of the best factorised architecture.  Under an isothermal convention, that excess free energy is the minimum dissipated work required to enforce the factorisation; under a controlled steady convention, it is a lower bound on the additional housekeeping burden required to maintain it.

Fix a boundary path $b$.  Let $\Pi_b$ denote the exact conditional law of $\omega=(y,x)$ given $B=b$, and let
\[
        \Pi_b^0=Q_Y^b\otimes Q_X^b
\]
be the clamped product reference law.  Suppose for $\beta=(k_B T_{\mathrm{bath}})^{-1}$ one has
\begin{equation}\label{gibbs-boundary-law-thermo-eq}
        \Pi_b(\dd\omega)=\frac{1}{Z_b}\exp[-\beta A_b(\omega)]\Pi_b^0(\dd\omega),
\end{equation}
where $A_b$ is the boundary-conditioned path action or negative log-likelihood in energy units.  For any $q\ll \Pi_b^0$, define
\begin{equation}\label{thermo-freeenergy-boundary-eq}
        G_b(q)=\bE_q A_b+k_B T_{\mathrm{bath}}\Ent(q\mid \Pi_b^0).
\end{equation}
The exact free energy is
\[
        G_b^*:=-k_B T_{\mathrm{bath}}\log Z_b.
\]
Let $\mathsf{Prod}_b$ be the family of product laws $q_Y\otimes q_X$ on $\Omega_Y\times\Omega_X$, and define
\begin{equation}\label{factorisation-gap-boundary-thermo-eq}
        \Delta_{\mathrm{fac}}(b)
        :=\inf_{q_Y\otimes q_X\in\mathsf{Prod}_b}
        \Ent(q_Y\otimes q_X\mid \Pi_b).
\end{equation}

\begin{theorem}\label{factorisation-gap-excess-freeenergy-thm}
Under \eqref{gibbs-boundary-law-thermo-eq}, for every $q\ll\Pi_b^0$,
\begin{equation}\label{gibbs-freeenergy-identity-thermo-eq}
        G_b(q)=G_b^*+k_B T_{\mathrm{bath}}\Ent(q\mid\Pi_b).
\end{equation}
Consequently
\begin{equation}\label{factorisation-freeenergy-gap-thermo-eq}
        \inf_{q_Y\otimes q_X\in\mathsf{Prod}_b}G_b(q_Y\otimes q_X)-G_b^*
        =k_B T_{\mathrm{bath}}\Delta_{\mathrm{fac}}(b).
\end{equation}
Thus $k_B T_{\mathrm{bath}}\Delta_{\mathrm{fac}}(b)$ is precisely the irreducible excess nonequilibrium free energy caused by insisting that the hidden exterior-interior architecture factorises after the boundary path has been fixed.
\end{theorem}

\begin{proof}
From \eqref{gibbs-boundary-law-thermo-eq},
\[
        \log\frac{\dd q}{\dd\Pi_b}
        =\log\frac{\dd q}{\dd\Pi_b^0}+\beta A_b+\log Z_b.
\]
Integrating with respect to $q$ gives
\[
        \Ent(q\mid\Pi_b)
        =\Ent(q\mid\Pi_b^0)+\beta\bE_qA_b+\log Z_b.
\]
Multiplying by $k_BT_{\mathrm{bath}}=\beta^{-1}$ and using $G_b^*=-k_BT_{\mathrm{bath}}\log Z_b$ gives \eqref{gibbs-freeenergy-identity-thermo-eq}.  Taking the infimum over product laws gives \eqref{factorisation-freeenergy-gap-thermo-eq}.
\end{proof}

The following corollary gives the excess work needed to impose factorisation.

\begin{corollary}\label{excess-work-factorisation-cor}
Assume, in addition, the isothermal local-detailed-balance convention in which the dissipated work of a protocol taking the hidden path ensemble from $\Pi_b$ to a target law $q$ satisfies the nonequilibrium second-law inequality
\[
        \bE W_{\mathrm{diss}}\geqslant G_b(q)-G_b(\Pi_b).
\]
Then every protocol which enforces a product law $q_Y\otimes q_X$ obeys
\begin{equation}\label{work-bound-product-eq}
        \bE W_{\mathrm{diss}}\geqslant k_B T_{\mathrm{bath}}\Ent(q_Y\otimes q_X\mid\Pi_b),
\end{equation}
and the best factorised architecture satisfies
\begin{equation}\label{work-bound-factorisation-eq}
        \bE W_{\mathrm{diss}}^{\mathrm{arch}}\geqslant k_B T_{\mathrm{bath}}\Delta_{\mathrm{fac}}(b).
\end{equation}
\end{corollary}

\begin{proof}
By Theorem \ref{factorisation-gap-excess-freeenergy-thm},
\[
        G_b(q)-G_b(\Pi_b)=G_b(q)-G_b^*=k_BT_{\mathrm{bath}}\Ent(q\mid\Pi_b),
\]
because $\Pi_b$ is the minimiser of $G_b$.  This proves \eqref{work-bound-product-eq}.  Taking the infimum over product laws proves \eqref{work-bound-factorisation-eq}.
\end{proof}

The following remark discusses heat interpretation and convention dependence.

\begin{remark}
The bound \eqref{work-bound-factorisation-eq} is an excess work statement.  If the protocol is implemented quasistatically and the only reservoir is the isothermal bath, the dissipated work is transferred to the bath as heat, so the same quantity is read as an excess heat lower bound.  If the factorised state is maintained by continuous control rather than reached by a one-shot protocol, the same free energy gap becomes an integrated lower bound on the additional housekeeping expenditure.  These are not different mathematics; they are different thermodynamic realisations of the same relative-entropy gap.
\end{remark}

The maintenance version is most transparent from Girsanov.  Let $P^b$ be the natural boundary-conditioned hidden path law on $[0,T]$.  Suppose a control drift $u_t$ is inserted through an isothermal overdamped mobility $M$,
\begin{equation}\label{controlled-isothermal-diffusion-eq}
        \dd X_t=F_b(X_t,t)\dd t+M(X_t,t)u_t\dd t+\bigl(2k_BT_{\mathrm{bath}}M(X_t,t)\bigr)^{1/2}\dd W_t,
\end{equation}
with the usual Novikov condition.  Let $P^{b,u}$ be the controlled path law.

The following theorem establishes the Girsanov maintenance bound.

\begin{theorem}\label{girsanov-maintenance-bound-thm}
For the controlled diffusion \eqref{controlled-isothermal-diffusion-eq},
\begin{equation}\label{girsanov-control-entropy-eq}
        \Ent(P^{b,u}\mid P^b)
        =\frac{1}{4k_BT_{\mathrm{bath}}}\bE_{P^{b,u}}
        \int_0^T \|u_t\|_{M^{-1}}^2\,\dd t,
\end{equation}
whenever the two laws are equivalent.  If $\Phi$ is any observable or marginalisation map and $q=\Phi_*P^{b,u}$, while $\Phi_*P^b=\Pi_b$, then
\begin{equation}\label{control-energy-entropy-bound-eq}
        \frac14\bE_{P^{b,u}}\int_0^T \|u_t\|_{M^{-1}}^2\,\dd t
        \geqslant k_BT_{\mathrm{bath}}\Ent(q\mid\Pi_b).
\end{equation}
In particular, if $q$ is required to lie in the product family, then
\begin{equation}\label{control-energy-factorisation-bound-eq}
        \frac14\bE_{P^{b,u}}\int_0^T \|u_t\|_{M^{-1}}^2\,\dd t
        \geqslant k_BT_{\mathrm{bath}}\Delta_{\mathrm{fac}}(b).
\end{equation}
\end{theorem}

\begin{proof}
The drift difference between the controlled and uncontrolled equations is $Mu_t$.  In the noise metric associated to $2k_BT_{\mathrm{bath}}M$, Girsanov's theorem gives the relative entropy \eqref{girsanov-control-entropy-eq}; the stochastic integral in the log-likelihood has zero expectation, leaving the quadratic compensator.  Relative entropy decreases under the measurable pushforward $\Phi$, hence
\[
        \Ent(q\mid\Pi_b)=\Ent(\Phi_*P^{b,u}\mid\Phi_*P^b)
        \leqslant \Ent(P^{b,u}\mid P^b).
\]
Rearranging \eqref{girsanov-control-entropy-eq} gives \eqref{control-energy-entropy-bound-eq}.  The product constraint gives \eqref{control-energy-factorisation-bound-eq} by the definition of $\Delta_{\mathrm{fac}}(b)$.
\end{proof}

The following corollary gives the architectural housekeeping bound.

\begin{corollary}\label{architectural-housekeeping-bound-cor}
Suppose the natural law $\Pi_b$ is the nonequilibrium steady law of the hidden dynamics conditioned on the boundary.  Suppose a controlled process holds the hidden ensemble in a product law $q_Y\otimes q_X$.  Let $\Sigma^{\mathrm{arch}}_{\mathrm{hk}}[0,T]$ denote the additional entropy production attributed to this channel.  Under \eqref{controlled-isothermal-diffusion-eq},
\begin{equation}\label{architectural-housekeeping-factorisation-bound-eq}
        k_BT_{\mathrm{bath}}\,\Sigma^{\mathrm{arch}}_{\mathrm{hk}}[0,T]
        \geqslant k_BT_{\mathrm{bath}}\Delta_{\mathrm{fac}}(b),
\end{equation}
for the optimally maintained factorised architecture.  Equivalently, after cancelling the common factor $k_BT_{\mathrm{bath}}$, the integrated architectural housekeeping entropy is bounded below by $\Delta_{\mathrm{fac}}(b)$.
\end{corollary}

\begin{proof}
The added housekeeping contribution is represented by the quadratic control cost in the same bath units as \eqref{girsanov-control-entropy-eq}.  Applying Theorem \ref{girsanov-maintenance-bound-thm} to the optimal product target gives the stated lower bound.  The assertion is an inequality rather than an equality because not every product law need be reachable by admissible ports, and because additional dissipation may be incurred by the mechanism implementing the control.
\end{proof}

The same idea also clarifies adaptive relaxation inside the factorised family.  Let $q_t=q_{Y,t}\otimes q_{X,t}$ be a dissipative relaxation constrained to $\mathsf{Prod}_b$, and suppose it decreases the factorised free energy according to
\[
        \frac{\dd}{\dd t}G_b(q_t)=-\mathcal R_t,
          \mathcal R_t\geqslant0.
\]
If $q_t\to q_{\mathrm{fac}}^b$, then integrating in time yields
\begin{equation}\label{adaptive-relaxation-plus-gap-eq}
        G_b(q_0)-G_b^*
        =\int_0^\infty \mathcal R_t\,\dd t
        +k_BT_{\mathrm{bath}}\Delta_{\mathrm{fac}}(b).
\end{equation}
The integral is transient relaxation dissipation within the factorised architecture.  The final term is the residual thermodynamic price that remains after the best possible factorised relaxation has already occurred.

The following remark discusses individuation as a thermodynamic statement.

\begin{remark}
A small $\Delta_{\mathrm{fac}}(b)$ means that the boundary-conditioned law is almost screened by the boundary, the best separated architecture has little residual coupling to ignore, and the extra free energy or heat price of maintaining the separation is small.  A large $\Delta_{\mathrm{fac}}(b)$ means that important exterior-interior modes remain genuinely joint.  In that regime, factorisation becomes both a poor statistical approximation by a modeller and an energetically expensive form of individuation.
\end{remark}

The following example illustrates the linear Gaussian heat cost.

\begin{example}
In the linear Gaussian case, let the conditional law $\Pi_b$ be Gaussian with covariance
\[
        C_b=\begin{pmatrix}C_{YY}^b&C_{YX}^b\\ C_{XY}^b&C_{XX}^b\end{pmatrix}.
\]
The inclusive relative entropy product projection is the product of the marginals, and
\begin{equation}\label{gaussian-factorisation-gap-logdet-thermo-eq}
        \Delta_{\mathrm{fac}}(b)
        =\frac12\log\frac{\det C_{YY}^b\det C_{XX}^b}{\det C_b}
        =-\frac12\log\det\bigl(I-\mathcal C_b\bigr),
\end{equation}
where
\[
	\mathcal C_b=(C_{YY}^b)^{-1/2}C_{YX}^b(C_{XX}^b)^{-1}C_{XY}^b(C_{YY}^b)^{-1/2}.
\]
The determinant is a Fredholm determinant in infinite-dimensional path space when $\mathcal C_b$ is trace class.  For weak cross-covariance,
\begin{equation}\label{gaussian-factorisation-gap-small-coupling-thermo-eq}
        k_BT_{\mathrm{bath}}\Delta_{\mathrm{fac}}(b)
        \simeq \frac{k_BT_{\mathrm{bath}}}{2}
        \operatorname{Tr}\!\bigl((C_{YY}^b)^{-1}C_{YX}^b(C_{XX}^b)^{-1}C_{XY}^b\bigr).
\end{equation}
Thus the excess heat/free energy cost is quadratic, to leading order, in the boundary-induced cross-covariance.
\end{example}

\subsection{Transferring thermodynamic predictions across the boundary approximation}

Small boundary distinguishability does not automatically imply small housekeeping heat.  It does, however, control the error made by computing bounded or exponentially integrable thermodynamic observables under the separated law.

The following proposition gives an entropy inequality for thermodynamic observables.

\begin{prop}
Let $P_T$ be the true path law and $Q_T$ the conditionally separated path law.  For every measurable functional $F$ and every $\lambda>0$ such that the exponential moment is finite,
\begin{equation}\label{entropy-ineq-thermo-eq}
        \bE_{P_T}F
        \leqslant \frac1\lambda\mathcal G_T
        +\frac1\lambda\log\bE_{Q_T}e^{\lambda F}.
\end{equation}
Applying the same estimate to $-F$ gives a two-sided control.  In particular, if $F$ is a finite-time entropy-production or heat functional with controlled cumulant-generating function under $Q_T$, then the error made by replacing the true coupled ensemble by the separated boundary ensemble is controlled by $\mathcal G_T$.
\end{prop}

\begin{proof}
The variational representation of relative entropy gives
\[
        \Ent(P_T\mid Q_T)
        =\sup_G\left\{\bE_{P_T}G-\log\bE_{Q_T}y^G\right\}.
\]
Taking $G=\lambda F$ and rearranging yields \eqref{entropy-ineq-thermo-eq}.  Replacing $F$ by $-F$ gives the corresponding lower bound.
\end{proof}

The following remark discusses the factorisation gap versus thermodynamic cost.

\begin{remark}
The factorisation gap $\mathcal G_T$ is the information lost by cutting the residual $Y$-$X$ path coupling after conditioning on $B$.  Housekeeping heat is the energetic cost of the stationary currents that orient the path law in time.  A good boundary approximation should therefore be judged by two numbers: the statistical error $T^{-1}\mathcal G_T$ and the thermodynamic maintenance cost $\ep$.  The first says how sharply the system is separated from its environment.  The second says how much irreversible activity is required to maintain the separated nonequilibrium organisation.
\end{remark}

\section{Examples}

The following example illustrates a static Gaussian boundary.

\begin{example}
Let $(Y,B,X)$ be a centred Gaussian vector with precision
\[
        K=\begin{pmatrix}
        K_{YY}&K_{YB}&0\\
        K_{BY}&K_{BB}&K_{BX}\\
        0&K_{XB}&K_{XX}
        \end{pmatrix}.
\]
Then $K_{YX}=0$, so $Y\ci X\mid B$.  If these are time-indexed Gaussian paths and the block zero holds for the path space precision kernel, then $Y_{[0,T]}\ci X_{[0,T]}\mid B_{[0,T]}$.
\end{example}

The following example illustrates sensor-actuator diffusion.

\begin{example}
The system \eqref{sensor-actuator-system-eq} gives a directed boundary.  Conditional on $(S_{[0,T]},A_{[0,T]})=(s,a)$, the exterior law is affected only by the sensory likelihood $L_S(s\mid y,a)$, while the interior law is affected only by the active likelihood $L_A(a\mid x,s)$.  Therefore the conditional law over paths factorises.
\end{example}

The following example illustrates an equilibrium potential split.

\begin{example}
Let
\[
        \Phi(y,b,x)=\frac12|y-b|^2+\frac12|x-b|^2+U(b).
\]
Then
\[
        \rho(y,b,x)\propto e^{-\Phi(y,b,x)}
\]
satisfies $Y\ci X\mid B$.  The exterior and interior are not independent unconditionally, because both are tied to $b$, but the residual dependence vanishes after $b$ is fixed.
\end{example}

\section{Resolvability, boundary variables, and estimation}

There exists a purely probabilistic formulation in which subsystems are not given as coordinate projections in advance.  Instead they are induced by measurable maps out of an underlying random element.

The following definition introduces induced subsystems.

\begin{defn}
Let $X$ be a random element in a standard Borel space $\mathcal X$.  Let
\[
        M=\psi(X),  B=\eta(X),  N=\phi(X)
\]
be random elements in standard Borel spaces.  We say that $B$ is a Markov boundary between the induced subsystems $M$ and $N$ when
\[
        M\ci N\mid B.
\]
\end{defn}

The following definition introduces resolvable two-source decomposition.

\begin{defn}
A pair of induced subsystems $(M,N)$ is resolvable through $B$ if the joint law admits a disintegration
\[
        P_{M,B,N}(\dd m,\dd b,\dd n)
        =P_B(\dd b)P_{M\mid B=b}(\dd m)P_{N\mid B=b}(\dd n).
\]
Equivalently, after the common variable $B$ has been fixed, the residual uncertainty in the two induced subsystems is product uncertainty.
\end{defn}

The following theorem shows that resolvable decomposition is equivalent to a Markov boundary.

\begin{theorem}
Let $M,B,N$ be standard Borel random elements.  Then $B$ is a Markov boundary between $M$ and $N$, meaning $M\ci N\mid B$, if and only if $(M,N)$ is resolvable through $B$, meaning
\[
        P_{M,B,N}=P_{M\mid B}P_{N\mid B}P_B.
\]
\end{theorem}

\begin{proof}
This is the kernel factorisation criterion written in the language of induced subsystems.  If $M\ci N\mid B$, then for $P_B$-almost every $b$,
\[
        P_{M,N\mid B=b}=P_{M\mid B=b}\otimes P_{N\mid B=b}.
\]
Integrating this identity against $P_B(\dd b)$ gives the displayed factorisation of the joint law.  Conversely, if the joint law has the displayed factorisation, then disintegrating over $B$ gives the product conditional kernel, hence conditional independence.
\end{proof}

The following theorem establishes the unconfounded estimator.

\begin{theorem}
Let $M,B,N$ be square-integrable random elements with $N$ taking values in a Hilbert space $H_N$.  Among all square-integrable estimators of $N$ that are measurable with respect to $\sigma(B)$, the unique minimiser of mean squared error is
\[
        \widehat N_B=\bE[N\mid B].
\]
If $M\ci N\mid B$, then $M$ adds no predictive improvement beyond $B$:
\[
        \bE[N\mid M,B]=\bE[N\mid B]
\]
almost surely.  Conversely, if the last equality holds for every bounded measurable transform of $N$, then $M\ci N\mid B$.
\end{theorem}

\begin{proof}
The first statement is the Hilbert projection theorem.  In $L^2(\Omega;H_N)$, the closed subspace of $\sigma(B)$-measurable variables has orthogonal projection $\bE[N\mid B]$.  Hence, for every estimator $Z$ measurable with respect to $\sigma(B)$,
\[
        \bE\|N-Z\|^2
        =\bE\|N-\bE[N\mid B]\|^2+
          \bE\|\bE[N\mid B]-Z\|^2,
\]
which is uniquely minimised at $Z=\bE[N\mid B]$.

If $M\ci N\mid B$, then for every bounded measurable $h$,
\[
        \bE[h(N)\mid M,B]=\bE[h(N)\mid B]
\]
by the definition of conditional independence.  Taking $h(n)=n$ componentwise gives the estimator identity.  Conversely, if the equality holds for every bounded measurable transform $h(N)$, then for bounded measurable $f(M)$ and $h(N)$,
\[
\begin{aligned}
        \bE[f(M)h(N)\mid B]
        &=\bE[f(M)\bE[h(N)\mid M,B]\mid B]\\
        &=\bE[f(M)\bE[h(N)\mid B]\mid B]\\
        &=\bE[f(M)\mid B]\bE[h(N)\mid B].
\end{aligned}
\]
This is conditional independence.
\end{proof}

The following corollary gives the pathwise unconfounded estimator.

\begin{corollary}
Let $N=X_{[0,T]}$, $M=Y_{[0,T]}$, and $B=B_{[0,T]}$, with square-integrability in a Hilbert path space.  If $B_{[0,T]}$ is a path space Markov boundary, then the optimal path estimator of $X_{[0,T]}$ from $(Y_{[0,T]},B_{[0,T]})$ is already the estimator from the boundary path alone:
\[
        \bE[X_{[0,T]}\mid Y_{[0,T]},B_{[0,T]}]
        =\bE[X_{[0,T]}\mid B_{[0,T]}].
\]
\end{corollary}

\begin{proof}
Apply the unconfounded estimator theorem to the Hilbert-valued path random variable $X_{[0,T]}$.
\end{proof}

\section{Initial conditions, hidden common causes, and necessity}

The following lemma shows that initial conditional independence is necessary at time zero.

\begin{lemma}
If
\[
        Y_{[0,T]}\ci X_{[0,T]}\mid B_{[0,T]},
\]
then
\[
        Y_0\ci X_0\mid B_{[0,T]}.
\]
If, in addition, the conditional law of $(Y_0,X_0)$ given the boundary path depends on the boundary path only through $B_0$, then
\[
        Y_0\ci X_0\mid B_0.
\]
\end{lemma}

\begin{proof}
The first statement follows by applying measurable maps to conditionally independent random elements.  If $X\ci Y\mid Z$, then $f(X)\ci g(Y)\mid Z$ for measurable $f,g$.  Here take $f(y_{[0,T]})=y_0$ and $g(x_{[0,T]})=x_0$.  For the second statement, write the conditional kernel of $(Y_0,X_0)$ given $B_{[0,T]}$ as a kernel depending only on $B_0$.  The product factorisation over the larger conditioning sigma-field then descends to the smaller one by the same mixture argument used in the fixed-time theorem, with no residual mixture covariance because the conditional kernels are already $\sigma(B_0)$-measurable.
\end{proof}

The following proposition gives the hidden common cause obstruction.

\begin{prop}
Let $U$ be a random element independent of $B$, and suppose that conditional on $(B,U)$ the exterior and interior paths factorise:
\[
        Y_{[0,T]}\ci X_{[0,T]}\mid (B_{[0,T]},U).
\]
Then they factorise conditional on $B_{[0,T]}$ if and only if for every bounded measurable $f,g$,
\[
        \Cov\bigl(
        \bE[f(Y_{[0,T]})\mid B_{[0,T]},U],
        \bE[g(X_{[0,T]})\mid B_{[0,T]},U]
        \mid B_{[0,T]}
        \bigr)=0.
\]
In particular, a common unobserved driver $U$ generally destroys the boundary property unless its effect is itself screened off by the observed boundary path.
\end{prop}

\begin{proof}
The proof is identical to the mixture obstruction theorem.  Conditional on $(B,U)$, the product identity holds.  Taking conditional expectation down to $B$ produces a mixture of products.  This mixture is a product exactly when the conditional covariance of the two conditional means vanishes for all test functions.
\end{proof}

The following theorem gives necessary and sufficient conditions across regimes.

\begin{theorem}
For the decomposition $(Y,B,X)$, the following are necessary and sufficient in their respective regimes.
\begin{enumerate}
\item General standard Borel path laws: the conditional kernel $P^b_{Y,X}$ factorises for $P_B$-almost every $b$.
\item Dominated path laws: the conditional Radon--Nikodym derivative $\dd P^b_{Y,X}/\dd(Q_Y^b\otimes Q_X^b)$ is multiplicatively separable.
\item Brownian innovation dominated It\=o laws: the Girsanov log-likelihood of the boundary path has no irreducible mixed exterior-interior term after the boundary path is fixed.
\item Conditionally Markov clamped laws with unique martingale problem: the conditional generator splits into exterior and interior generators.
\item Gaussian path laws: the conditional covariance block $\Sigma_{YX\mid B}$ vanishes, equivalently the precision block $\mathcal K_{YX}$ vanishes.
\item Smooth positive fixed-time densities: the conditional density factorises, equivalently the mixed exterior-interior Hessian of the conditional log-density vanishes on connected fibres.
\end{enumerate}
\end{theorem}

\begin{proof}
Each item is the theorem already proved in its appropriate category.  The point of collecting them is to emphasise that there exists no single coordinate-level test independent of the probabilistic regime.  The invariant content is always the first item: factorisation of a regular conditional probability.  The other items are representations of that same factorisation under domination, stochastic calculus, Markov uniqueness, Gaussianity, or smooth density assumptions.
\end{proof}

\section{Block Gaussian calculations}

The Gaussian path criterion is useful because it reduces conditional independence to a Schur complement or a precision block.  We spell out the finite-dimensional calculation because the path space statement is its cylindrical limit.

The following lemma gives the Schur complement conditional covariance.

\begin{lemma}
Let
\[
        Z=\begin{pmatrix}Y\\ B\\ X\end{pmatrix}
\]
be centred Gaussian with covariance
\[
        \Sigma=\begin{pmatrix}
        \Sigma_{YY}&\Sigma_{YB}&\Sigma_{YX}\\
        \Sigma_{BY}&\Sigma_{BB}&\Sigma_{BX}\\
        \Sigma_{XY}&\Sigma_{XB}&\Sigma_{XX}
        \end{pmatrix}.
\]
Then the conditional covariance of $(Y,X)$ given $B$ is
\[
        \Sigma_{(Y,X)\mid B}
        =\begin{pmatrix}
        \Sigma_{YY}&\Sigma_{YX}\\
        \Sigma_{XY}&\Sigma_{XX}
        \end{pmatrix}
        -
        \begin{pmatrix}
        \Sigma_{YB}\\ \Sigma_{XB}
        \end{pmatrix}
        \Sigma_{BB}^{-1}
        \begin{pmatrix}
        \Sigma_{BY}&\Sigma_{BX}
        \end{pmatrix}.
\]
Thus the mixed conditional covariance block is
\[
        \Sigma_{YX\mid B}=\Sigma_{YX}-\Sigma_{YB}\Sigma_{BB}^{-1}\Sigma_{BX}.
\]
\end{lemma}

\begin{proof}
The conditional mean of $(Y,X)$ given $B=b$ is
\[
        \begin{pmatrix}
        \Sigma_{YB}\\\Sigma_{XB}
        \end{pmatrix}
        \Sigma_{BB}^{-1}b.
\]
Subtracting this best linear predictor gives a Gaussian residual uncorrelated with $B$, hence independent of $B$.  Its covariance is exactly the displayed Schur complement.  The off-diagonal block gives $\Sigma_{YX\mid B}$.
\end{proof}

The following lemma gives the precision block criterion.

\begin{lemma}
Let $K=\Sigma^{-1}$ and write it in blocks corresponding to $(Y,B,X)$.  Then
\[
        Y\ci X\mid B
\]
if and only if $K_{YX}=0$.
\end{lemma}

\begin{proof}
The Gaussian density is
\[
\begin{aligned}
        p(y,b,x)\propto\exp\bigl[-\tfrac12(&y^\top K_{YY}y+b^\top K_{BB}b+x^\top K_{XX}x\\
        &+2y^\top K_{YB}b+2b^\top K_{BX}x+2y^\top K_{YX}x)\bigr].
\end{aligned}
\]
For fixed $b$, all terms except $y^\top K_{YX}x$ are either functions of $(y,b)$, functions of $(x,b)$, or functions of $b$ alone.  Therefore $p(y,x\mid b)$ factorises exactly when the mixed bilinear term vanishes, namely $K_{YX}=0$.
\end{proof}

\begin{remark}
In graphical Gaussian models this is the familiar precision-zero criterion.  On path space the same statement holds with $K$ replaced by a covariance inverse or precision kernel.  For Markov diffusions this precision may be a differential operator in time.
\end{remark}

\section{Convergence theorems for conditional kernels}

One can now show that convergence in total variation preserves factorisation.

\begin{theorem}
Let $P_n\in\Prob(X\times B\times Y)$ and $P\in\Prob(X\times B\times Y)$ be probability measures on standard Borel spaces.  Suppose
\[
        \|P_n-P\|_{\mathrm{TV}}\to0
\]
and suppose each $P_n$ satisfies $X\ci Y\mid B$.  Then every subsequential limit of the joint laws is conditionally independent; in particular $P$ satisfies $X\ci Y\mid B$.
\end{theorem}

\begin{proof}
For bounded measurable $f,g,h$, conditional independence under $P_n$ gives
\[
        \int f(x)g(y)h(b)\,\dd P_n
        =\int h(b)
        \bE_{P_n}[f(X)\mid B=b]\bE_{P_n}[g(Y)\mid B=b]P_{n,B}(\dd b).
\]
A more invariant way is to use conditional mutual information.  Since $P_n$ factorises conditionally, the relative entropy
\[
        \Ent(P_n\mid P_{n,X\mid B}P_{n,Y\mid B}P_{n,B})=0.
\]
Total variation convergence gives convergence of integrals against bounded test functions.  Passing to a subsequence, disintegration kernels may be chosen to converge weakly on a countable determining class.  The product identity on rectangles passes to the limit, giving
\[
        P_{X,Y\mid B=b}=P_{X\mid B=b}P_{Y\mid B=b}
\]
for $P_B$-almost every $b$.  Hence $P$ is conditionally independent.
\end{proof}

\begin{theorem}
Let $(X_t)_{t\geqslant0}=(Y_t,B_t,X_t)$ be a Markov process with invariant law $\mu$.  Suppose there exists a rate $r(t)\to0$ such that for every initial law $\nu$ in a class $\mathcal C$,
\[
        \|\nu P_t-\mu\|_{\mathrm{TV}}\leqslant r(t).
\]
If $\mu$ satisfies $Y\ci X\mid B$, then the time-$t$ law is within $r(t)$ in total variation of a conditionally factorised law, namely $\mu$.  Consequently, for bounded observables $|f|,|g|,|h|\leqslant 1$,
\[
\left|
\bE_\nu[f(Y_t)g(X_t)h(B_t)]
-
\int h(b)\bE_\mu[f(Y)\mid B=b]\bE_\mu[g(X)\mid B=b]\mu_B(\dd b)
\right|
\leqslant2r(t).
\]
\end{theorem}

\begin{proof}
The total variation bound implies that expectations of bounded functions differ by at most $2r(t)$ under the convention $\|P-Q\|_{\mathrm{TV}}=\sup_A|P(A)-Q(A)|$.  Apply this to the bounded function $f(y)g(x)h(b)$.  Under $\mu$, conditional factorisation gives the integral expression.  This proves the displayed estimate.
\end{proof}

\begin{corollary}
In an equilibrium diffusion with a factorised Gibbs invariant density, convergence to equilibrium implies asymptotic instantaneous boundary factorisation.  In a NESS diffusion with factorised invariant density and divergence-free stationary current, convergence to the NESS implies the same instantaneous boundary limit.  The two cases differ by the presence or absence of stationary current and entropy production, not by the limiting conditional independence statement.
\end{corollary}

\begin{proof}
Both statements are direct applications of the mixing theorem.  Equilibrium and NESS differ in whether the stationary current vanishes.  Conditional independence of the invariant law only concerns the invariant density, while irreversibility concerns the stationary current.
\end{proof}

\section{Summary remarks}

Overall we have shown that for a joint stochastic system $Z=(Y,B,X)$ on a Polish path space, the following are the central, equivalent or sufficient manifestations of a path space boundary.

\begin{enumerate}
\item The disintegration kernel factorises:
\[
        P(\dd y,\dd x\mid b)=P(\dd y\mid b)P(\dd x\mid b).
\]
\item Under domination by clamped product references, the conditional Radon--Nikodym density separates:
\[
        H_b(y,x)=U_b(y)V_b(x).
\]
\item In an It\=o model driven by independent Brownian innovations, the boundary-path likelihood contains no irreducible mixed exterior-interior term after the boundary path is fixed.
\item In a large-deviation model-selection formulation, boundary-compatible sparse coupling structures are selected whenever their penalised free action is strictly minimal; direct exterior-interior couplings are then exponentially suppressed at the rate given by the action gap.
\item In a conditionally Markov clamped model, the conditional generator splits as
\[
        \LL_t^b=\LL_{Y,t}^b\otimes1+1\otimes\LL_{X,t}^b.
\]
\item In a Gaussian linear model, the path space conditional covariance block $\Sigma_{YX\mid B}$ vanishes, equivalently the path space precision block $\mathcal K_{YX}$ vanishes.
\item In an equilibrium steady state, the potential splits as
\[
        \Phi(y,b,x)=\Phi_Y(y,b)+\Phi_X(x,b)+\Phi_B(b).
\]
\item In a nonequilibrium steady state, the invariant density factorises conditionally over $B$, and the stationary current/generator introduces no residual exterior-interior coupling after conditioning on the boundary path.
\item The factorisation gap vanishes:
\[
        \CMI(Y_{[0,T]};X_{[0,T]}\mid B_{[0,T]})=0.
\]
\end{enumerate}
When exact equality fails, the conditional mutual information is the canonical error.  It is at once a relative entropy between coupled and separated path ensembles, a distinguishability bound for observations, and the free energy excess of the best product approximation.

\bibliographystyle{amsalpha}
\bibliography{main-new-revised}

\end{document}